\documentclass[a4paper, USenglish, cleveref, autoref, thm-restate, numberwithinsect]{lipics-v2021}
\usepackage{craigstyle}
\usepackage{xspace}
\usepackage{amssymb}
\usepackage[T1]{fontenc}
\usepackage{graphicx}
\usepackage{ifthen}
\usepackage{hyperref}
\usepackage{stmaryrd}
\usepackage{thm-restate}
\usepackage{tikz}
\usepackage{verbatim}
\usepackage{bussproofs} \EnableBpAbbreviations

\newcommand{\nnf}{\textrm{\sc nnf}}
\newcommand{\subf}{\textrm{\sc subf}}
\newcommand{\lits}{\textrm{\sc literals}}
\newcommand{\nf}{\mathop{\textup{nf}}}

\newcommand{\logic}[1]{\ensuremath{\mathsf{#1}}}
\newcommand{\frameclass}[1]{\ensuremath{\mathcal{#1}}}
\newcommand{\pair}[1]{(#1)}

\newcommand{\faThumbsOUp}[1][\!]{\includegraphics[scale=.02]{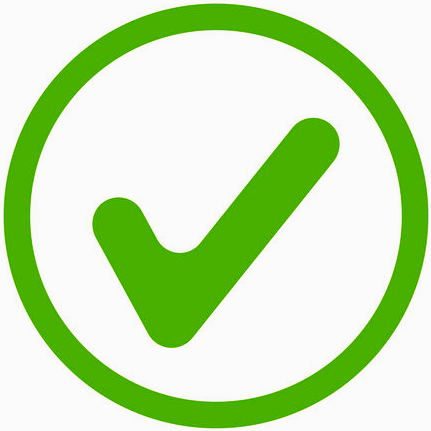}#1}
\newcommand{\faThumbsODown}[1][\!]{\includegraphics[scale=.02]{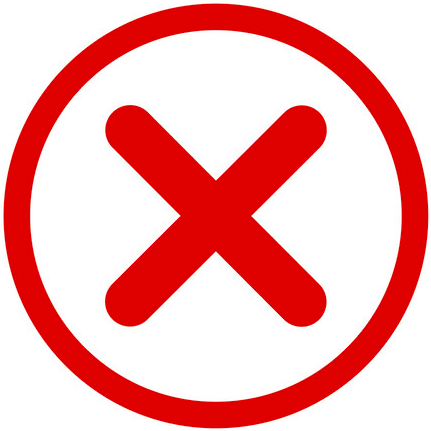}#1}
\newcommand{\pro}[1]{\begin{itemize}\item[\faThumbsOUp] #1\end{itemize}}
\newcommand{\con}[1]{\begin{itemize}\item[\faThumbsODown] #1\end{itemize}}

\newenvironment{mainproof}{\medskip\par\noindent\textbf{Proof of Craig interpolation for modal logic.}}{\qed\medskip\par\noindent}

\newcommand{\sig}{\mathop{\textup{sig}}}

\newcommand{\FO}{\textup{FO}\xspace}

\newcommand{\ST}{\textup{ST}}
\newcommand{\eoe}{\lipicsEnd}

\usepackage{tikz}
\usepackage{tikz-cd}
\usetikzlibrary{shapes.geometric}
\tikzset{
itria/.style={
  draw,shape border uses incircle,
  isosceles triangle,shape border rotate=90, yshift = -1.6cm, xshift = 0.7cm },
rtria/.style={
  draw,shape border uses incircle, minimum width = 1cm,
  minimum height = 1cm,
  isosceles triangle,isosceles triangle apex angle=130,
  shape border rotate=-65,yshift=-0.098cm,xshift=0.9346cm},
ritria/.style={
  draw,shape border uses incircle,
  isosceles triangle,isosceles triangle apex angle=110,
  shape border rotate=-55,yshift=0.1cm},
letria/.style={
  draw,dashed,shape border uses incircle,
  isosceles triangle,isosceles triangle apex angle=110,
  shape border rotate=235,yshift=0.1cm}
  
}

\newcommand{\bisim}{\mathop{{\underline\leftrightarrow}}}

\newcommand{\CPC}{{\sf CPC}}
\newcommand{\K}{{\sf K}}

\newcommand{\G}{{\sf G}} 
\newcommand{\Gth}{{\sf G3}}
\newcommand{\Gthp}{{\sf G3p}}
\newcommand{\GthK}{{\sf G3K}}

\newcommand{\imp}{\rightarrow}

\newcommand{\en}{\wedge} 
\newcommand{\of}{\vee}

\newcommand{\bx}{\raisebox{0mm}{$\Box$}}

\newcommand{\fnbx}{\raisebox{0mm}{{\footnotesize $\Box$}}}
\newcommand{\dm}{\raisebox{-.35mm}{$\Diamond$}}
\newcommand{\bof}{\bigvee}
\newcommand{\ben}{\bigwedge}
\newcommand{\sgn}{\text{sgn}}

\newcommand{\De}{\Delta}
\newcommand{\Ga}{\Gamma}

\newcommand{\Sig}{\Sigma}

\renewcommand{\phi}{\varphi}

\newcommand{\varchi}{\raisebox{.4ex}{\mbox{$\chi$}}}

\newcommand{\seq}{\Rightarrow}

\newcommand{\defn}{\equiv _{\mbox{\em \tiny df}}} 
\newcommand{\af}{\vdash}

\newcommand{\form}{{\EuScript F}}
\newcommand{\parseq}{{\EuScript Ps}}

\author{Nick Bezhanishvili}{University of Amsterdam}{n.bezhanishvili@uva.nl}{https://orcid.org/0009-0005-6692-5051}{Acks}
\author{Balder ten Cate}{University of Amsterdam}{b.d.tencate@uva.nl}{https://orcid.org/0000-0002-2538-5846}{Supported by the European Union’s Horizon 2020 research and innovation programme (MSCA-101031081).}
\author{Rosalie Iemhoff}{Utrecht University}{r.iemhoff@uu.nl}{https://orcid.org/0000-0001-9975-9604}{Support by the Netherlands Organisation for Scientific Research under grant 639.073.807 as well as partial support by the MOSAIC project (EU H2020-MSCA-RISE-2020 Project 101007627) is gratefully acknowledged.} 

\newtheorem{question}{Open Question}

\begin{document}

\title{Six Proofs of Interpolation for the Modal Logic \logic{K}}
%
%

\maketitle

\begin{abstract}
In this chapter, we present six different proofs of Craig interpolation 
for the modal logic $\mathsf{K}$, each using a different set of techniques (model-theoretic, proof-theoretic, syntactic, automata-theoretic, using quasi-models, and algebraic). We 
compare the pros and cons of each proof technique.
\end{abstract}

\setcounter{tocdepth}{1}
\tableofcontents

\section{Introduction}

A logic has the Craig interpolation property if, whenever an implication $\phi\to\psi$ is valid,
there is a formula $\vartheta$ (which we will call an ``interpolant'') such that 
$\phi\to\vartheta$ and $\vartheta\to\psi$ are valid, and such that all the non-logical symbols occurring in $\vartheta$ occur both in $\phi$ and in $\psi$. Craig~\cite{Craig1957} proved that first-order logic has this property. This result was later refined by Lyndon~\cite{Lyndon1959}, 
who proved that every valid implication has an interpolant, such that
every non-logical symbol occurring positively in the interpolant 
occurs positively in both the antecedent and the consequent, and likewise for negative occurrences. Several further strengthenings of these interpolation theorems  were later obtained for first-order logic. See \refchapter{chapter:predicate} for more details.
Many modal logics also have Craig interpolation, and indeed, Lyndon interpolation~\cite{Gabbay1972:craig}. Some (such as the basic modal logic \logic{K}) even enjoy a strong form of interpolation that fails for first-order logic, called \emph{uniform interpolation}, whereby the interpolant can be constructed in a uniform way, such that it does not depend on the consequent, but only on the antecedent and shared language (i.e., the set of non-logical symbols shared by the antecedent and consequent).

\refchapter{chapter:nonclassical} provides a broader picture regarding interpolation properties for non-classical logics and the landscape of modal logics with Craig interpolation.
In this chapter, we take a close look at one specific interpolation theorem, namely the  Craig interpolation theorem for the  modal logic \logic{K}:

\begin{theorem}[Craig Interpolation for the Modal Logic \logic{K}]\label{thm:modal-cip}
    Every valid modal implication $\models\phi\to\psi$ has a \emph{Craig interpolant}, that is, a modal formula $\vartheta$ such that
    the following hold:
    \begin{enumerate}
        \item $\models\phi\to\vartheta$,
        \item $\models\vartheta\to\psi$,
        \item $\sig(\vartheta)\subseteq \sig(\phi)\cap \sig(\psi)$,
    \end{enumerate}
\end{theorem}

Note that this version of Craig interpolation for modal logic is sometimes called \emph{arrow-interpolation} (also known as \emph{local interpolation}), as opposed to \emph{turnstile-interpolation} (also known as \emph{global interpolation} or \emph{deductive interpolation}), cf.~\refchapter{chapter:nonclassical}.

There are many ways to prove this theorem. In this chapter we  review six proofs. Each draws on a different set of techniques
(model theoretic, proof theoretic, syntactic, automata-theoretic,  using quasi-models, and algebraic), 
and each  proof has different advantages and disadvantages. 
For example, some techniques generalize better to other logics, some techniques allow proving stronger forms of interpolation such as \emph{uniform interpolation}, and some techniques yield an algorithm for producing interpolants of (worst-case) \emph{optimal size}. Thus, different techniques are useful for different purposes. 
In addition, we hope that the different proofs of interpolation  provide the reader with a glimpse at the different angles from which modal logics can be studied more generally.

\subparagraph*{A note on formula size.}
One of the questions we will be interested in is
whether the different proofs of Craig interpolation yield algorithms for computing interpolants, and,
if so, whether they produce small interpolants.
Here, it is important to point out that the \emph{size} of a modal 
formula can be measured in different ways. The  standard definition of formula length is given by the number of 
symbols needed to write the formula as a string.
We denote this by $|\phi|$. The other  notion of size is as the cardinality of the set of its subformulas. We denote this by $|\phi|_{\text{DAG}}$. It captures the size of the formula when represented succinctly as a directed acyclic graph (DAG). In general,
$|\phi|_{\text{DAG}}$ can be exponentially smaller than $|\phi|$ \cite{Ditmarsch2014:exponential,Berkholz24:modal}. This is illustrated by the following example.

\begin{example}
\label{ex:ditmarsch}
Consider the modal formula $\chi_n$ defined
inductively as follows:
\[\begin{array}{lll}
\chi_0     &=& (\Box p \lor \Box \neg p) \\
\chi_{i+1} &=& (\Box\chi_i \lor \Box \neg\chi_i)
\end{array}\]
In other words, $\chi_n = \blacksquare^n p$, where
$\blacksquare\phi$ is short for the ``non-contingency statement'' $\Box\phi\lor\Box\neg\phi$. Observe that $|\chi_n|_{\text{DAG}} = 4n+6$ whereas $|\chi| = 14 \cdot 2^n-6$ (if we include the parentheses).
In fact, it has been shown that every modal
formula $\psi$ equivalent to $\chi_n$ has size $|\psi| \geq 2^n$ \cite{Ditmarsch2014:exponential}. A similar super-polynomial gap between $|\phi|$
and $|\phi|_{\text{DAG}}$ is believed to hold already for propositional formulas, but proving it would amount to separating the complexity class $\textsf{NC}_1$ from \textsf{P/poly}, an open problem in complexity theory~\cite{Rossman2018:formulas}. See 
\refchapter{chapter:propositional} for more details.
\eoe
\end{example}

As we will see, some of the approaches yield an algorithm for computing a Craig interpolant of singly exponential size, even under the standard formula size metric ($|\phi|$), while others
yield an algorithm for computing a Craig interpolant of singly exponential size under the 
more permissive DAG-based size metric ($|\phi|_{\text{DAG}}$).
We will discuss matching lower bounds on Craig interpolant size  in Section~\ref{sec:lowerbounds}.

\subparagraph*{Outline}
In Section~\ref{sec:basic} we review basic definitions. In the subsequent six sections,
we review different proofs of Theorem~\ref{thm:modal-cip}. 
In Section~\ref{sec:lowerbounds} we discuss some lower bound results for computing Craig interpolants.

\section{Basic Definitions}
\label{sec:basic}

Although we assume familiarity with modal logics, we briefly recall some basic definitions 
in order to fix our notation. See~\cite{BlackburnDeRijkeVenema}
for more details.

\subsection*{Modal Formulas: Syntax and Kripke Semantics}
By a \emph{propositional signature} we will mean a set of proposition letters (also known as \emph{propositional variables}).  Given a propositional signature $\tau$, the \emph{modal formulas} w.r.t. $\tau$ are generated by the following
grammar:
\[ \chi ::= p\mid \top\mid \bot\mid \neg\chi \mid (\chi\lor\chi) \mid (\chi\land\chi) \mid \Diamond\chi\mid \Box\chi\]
where $p\in\tau$. Note that it will be important that we allow $\top$ and $\bot$ as atomic formulas. We will view the
implication $\phi\to\psi$ as a shorthand notation for 
$\neg\phi\lor\psi$. 
We will denote by $\sig(\chi)$ the propositional signature of $\chi$, that is, the set of proposition letters occurring in $\chi$.

     A modal formula is in \emph{negation normal form (NNF)} if every negation sign is immediately in front of a proposition letter. In other words, the modal formulas in NNF are generated by the grammar $\chi ::= p \mid \neg p \mid \top \mid \bot \mid \chi_1\land\chi_2\mid 
       \chi_1\lor\chi_2 \mid \Diamond \chi \mid \Box\chi$.
     We denote by $\nnf(\chi)$ the result of bringing a formula $\chi$ into NNF (by using the De Morgan laws and the duality laws $\neg\Diamond\chi \equiv \Box\neg\chi$ and $\neg\Box\chi\equiv \Diamond\neg\chi$).

A \emph{Kripke frame} is a pair $F=(W,R)$ where $W$ is a set of \emph{worlds} and $R\subseteq W\times W$. A \emph{Kripke model}
for a propositional signature $\tau$
is a triple $M=(W,R,V)$ where $(W,R)$ is a Kripke frame (the \emph{underlying frame} of $M$), and $V:\tau\to\wp(W)$ assigns to every proposition letter a set of worlds. Satisfaction of a 
modal formula at a world in a Kripke model, denoted
$M,w\models\chi$, is defined as usual, cf.~Table~\ref{tab:semantics}.
For a set of formulas $\Phi$, 
we will write $M,w\models\Phi$
if $M,w\models\phi$
for all $\phi\in \Phi$.

We will also refer to a pair $(M,w)$,
where $M=(W,R,V)$ is a Kripke model and $w\in W$ a world, as a 
\emph{pointed Kripke model}. Thus,  satisfaction can be viewed as a relation between pointed Kripke models and modal formulas. 

\begin{table}[t]
\[\begin{array}{lcl}
M,w\models p &\text{iff}& w\in V(p) \\
M,w\models \top &\text{always} \\
M,w\models \bot &\text{never} \\
M,w\models \neg\chi &\text{iff}& M,w\not\models\chi\\
M,w\models \chi_1\land\chi_2 &\text{iff}& M,w\models\chi_1 \text{ and } M,w\models\chi_2\\
M,w\models \chi_1\lor\chi_2 &\text{iff}& M,w\models\chi_1 \text{ or } M,w\models\chi_2\\
M,w\models \Diamond\chi &\text{iff}& M,w'\models\chi \text{ for some $w'\in W$ with $(w,w')\in R$ }\\
M,w\models \Box\chi &\text{iff}& M,w'\models\chi \text{ for all $w'\in W$ with $(w,w')\in R$ }\\
\end{array}
\]
\caption{Kripke semantics for modal logic}
\label{tab:semantics}
\end{table}

\subsection*{Validity, Frame Classes and Logics}

We write $\models\chi$ if the modal formula $\chi$ is
satisfied at every world in every Kripke model.
This can be further relativized to a class of Kripke frames (or, \emph{frame class} for short).
Let $\frameclass{F}$ be any frame class. Then we write
$\frameclass{F}\models\phi$ if $\phi$ is satisfied at every 
world in every Kripke model whose underlying frame belongs to $\frameclass{F}$. 
Every frame class induces  a \emph{modal logic}. Formally, 
fix a countably infinite propositional signature $\tau$. Then
by a \emph{normal modal logic} we will mean a set $L$ of modal formulas
that includes all substitution instances of propositional tautologies as well as all formulas of the form $\Box(\chi_1\to\chi_2)\to(\Box\chi_1\to\Box\chi_2)$ 
and is closed under modus ponens (if $\chi_1\to\chi_2\in L$ and $\chi_1\in L$, then $\chi_2\in L$), necessitation (if $\chi\in L$, then $\Box\chi\in L$), and substitution (if $\chi\in L$, then every substitution instance of $\chi$ belongs to $L$).
For every frame class $\frameclass{F}$,
the set of modal formulas $\{\phi\mid \frameclass{F}\models\phi\}$
is, in this sense, a modal logic. We will refer to it as \emph{the modal logic of $\frameclass{F}$}.
The modal logic of the class of all frames
is known as the modal logic \logic{K}, while, for example,
the modal logic of the class of transitive frames
is known as the modal logic \logic{K4}.

\subsection*{Modal Logic as a Fragment of First-Order Logic}

\begin{table}[t]
\[\begin{array}{lcl}
\ST_x(p) &=& P(x) \\
\ST_x(\top) &=& \top \\
\ST_x(\bot) &=& \bot \\
\ST_x(\neg\chi) &=& \neg \ST_x(\chi)\\
\ST_x(\chi_1\land\chi_2) &=& \ST_x(\chi_1)\land \ST_x(\chi_2) \\
\ST_x(\chi_1\lor\chi_2) &=& \ST_x(\chi_1)\lor \ST_x(\chi_2)\\
\ST_x(\Diamond\chi) &=& \exists y(R(x,y)\land \ST_y(\chi))\\
\ST_x(\Box\chi) &=& \forall y(R(x,y)\to \ST_y(\chi))\\
\end{array}
\]
\caption{Standard translation $\ST_x(\cdot)$ from modal logic to first-order logic}
\label{tab:ST}
\end{table}

A Kripke model with propositional signature $\tau$ can be viewed as a  first-order structure over a
corresponding relational first-order signature $\sigma$. More precisely, $\sigma$ consists of
a single binary relation symbol $R$ plus, for each $p\in \tau$, a corresponding unary relation symbol $P$.  
It should be clear how every Kripke model
with propositional signature $\tau$ can be viewed as 
a first-order structure over the corresponding first-order signature $\sigma$. 
This correspondence between Kripke models and
first-order structures is, in fact, so basic and transparent that we will take it for granted and we will allow ourselves to freely treat Kripke models as first-order structures. This will allow us, for instance, to interpret first-order formulas (over the correspondence signature) in Kripke models. Following this
perspective, modal logic can be viewed as a fragment of first-order logic, in the following
sense: for every modal formula $\phi$ there is
a first-order formula $\psi(x)$ such that, 
for all pointed Kripke models $(M,w)$, 
$M,w\models\phi$ iff $M,w\models\psi$. Here, 
by the latter, we mean that $M$ satisfies $\psi(x)$ under the assignment that sends $x$ to $w$. This first-order equivalent $\psi(x)$ of a modal formula $\phi$ can be obtained
from $\phi$ by means of a simple, inductive, translation, which is known as the \emph{standard translation}. It is given in~Table~\ref{tab:ST}.

\begin{theorem}For all pointed Kripke models $(M,w)$ and modal formulas $\phi$, $M,w\models\phi$ iff $M,w\models\ST_x(\phi)$.
\end{theorem}

There is no such translation from first-order logic back to modal logic. For example, 
the first-order formula $\phi(x)=\exists y(R(x,y)\land R(y,y)\land P(y))$ is 
not equivalent to any modal formula. 
This can be proved formally using \emph{bisimulations}. 

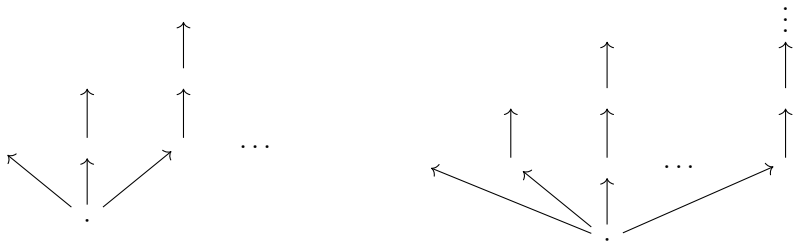
\begin{figure}[t]
\[
\begin{tikzcd}
      &              & ~            &  \\
      & ~            & \arrow[u,""] &  \\
~     & \arrow[u,""] & \arrow[u,""] & \!\!\!\!\!\!\!\!\ldots \\
      & \cdot \arrow[lu,""] \arrow[u,""] \arrow[ru,""]      
\end{tikzcd}
~~~~~~~~~~~~~
\begin{tikzcd}
      &              & ~            &        & \vdots \\
      & ~            & \arrow[u,""] &        & \arrow[u,""]\\
~     & \arrow[u,""] & \arrow[u,""] & \!\!\!\!\!\!\!\!\ldots & \arrow[u,""] \\
      && \cdot \arrow[llu,""] \arrow[lu,""] \arrow[u,""] \arrow[rru,""]     
\end{tikzcd}
\]
    \caption{Pointed Kripke models (over the empty propositional signature) that are
    modally indistinguishable but not bisimilar. The first has  branches of all finite lengths. The second has, in addition, an infinite branch.}
    \label{fig:not-modally-saturated}
\end{figure}

Given two Kripke models $M=(W,R,V), M'=(W',R',V')$ for the same propositional signature $\tau$, a \emph{$\tau$-bisimulation} between $M$ and $M'$ is a binary relation $Z\subseteq W\times W'$ such that the following conditions hold for all $(w,w')\in Z$:
\begin{description}
    \item[Atomic Harmony] For all $p\in \tau$, $w\in V(p)$ iff $w'\in V'(p)$,
    \item[Forth] For each $(w,v)\in R$ there is a $(w',v')\in R'$
            such that $(v,v')\in Z$,
    \item[Back] For each $(w',v')\in R'$ there is a $(w,v)\in R$
            such that $(v,v')\in Z$.
\end{description}
A $\tau$-bisimulation $Z$ between $M=(W,R,V)$ and $M'=(W',R',V')$ is said to be \emph{full} if 
for each $w\in W$ there is a $w'\in W'$ with $(w,w')\in Z$,
and, conversely, for each $w'\in W'$ there is  a
$w\in W$ with $(w,w')\in Z$.
We write $(M,w)\bisim_\tau (M',w')$, and we say that $(M,w)$ is 
\emph{$\tau$-bisimilar} to $(M',w')$, if there is a $\tau$-bisimulation
$Z$ between $M$ and $M'$ with $(w,w')\in Z$.

\begin{example} 
\label{ex:bisimulation}
Let $\tau=\{p\}$ and consider the following  
pointed Kripke models: 
\[
\begin{tikzcd}[row sep = 10mm]
    M: & w_0 \arrow[r] & 
    w_1^{~p} \arrow[loop right]{u} \\
    N: & v_0 \arrow[r] \arrow[u, dotted, dash]&
    v_1^{~p} \arrow[r] \arrow[u, dotted, dash]& 
    v_2^{~p} \arrow[ul, dotted, dash] \arrow[loop right]{u} 
\end{tikzcd}
\]
Then $(M,w_0)$ and $(N,v_0)$
are $\tau$-bisimilar. As bisimulation we can 
pick the binary relation $Z=\{(w_0,v_0), (w_1,v_1), (w_1,v_2)\}$ as indicated by the dotted lines. 
\eoe
\end{example}

It is a well-known fact that modal formulas are bisimulation-invariant. 

\begin{theorem}
\label{thm:bisim-modal-equiv}
Let
$(M,w)$ and $(M',w')$ be pointed Kripke models
with 
$(M,w)\bisim_\tau (M',w')$ for some propositional signature $\tau$. Then for all modal formulas $\chi$ with
$\sig(\chi)\subseteq\tau$, we have that
 $M,w\models\chi$ iff
 $M',w'\models\chi$.
\end{theorem}

In other words, if two pointed Kripke models $(M,w)$ and $(M',w')$ are $\tau$-bisimilar, then they are
\emph{modally indistinguishable w.r.t.\ $\tau$}, 
by which we mean that they satisfy the same
modal formulas in the propositional signature $\tau$.

\begin{example}
    We can use Theorem~\ref{thm:bisim-modal-equiv}, in combination with the Kripke models from Example~\ref{ex:bisimulation}, to show that the \FO-formula $\phi(x)=\exists y (R(x,y)\land R(y,y)\land P(y))$ is not equivalent to (the standard translation of) any modal formula in the propositional signature $\tau=\{p\}$. Indeed, 
    if there were such a modal formula, it would
    have to be true in $(M,w_0)$ but false in  $(N,v_0)$. But this is impossible, as
    the two are $\tau$-bisimilar.
    \eoe
\end{example}

The converse of Theorem~\ref{thm:bisim-modal-equiv} does not hold: modal indistinguishability does not imply the existence of a bisimulation.\footnote{By the Hennessy-Milner Theorem, the converse does hold for ``image-finite'' models, i.e., Kripke models in which each world has at most finitely many successors.}
Figure~\ref{fig:not-modally-saturated} depicts a well-known example of pointed Kripke models that are modally indistinguishable w.r.t.~$\tau$ but not $\tau$-bisimilar (where we can even pick $\tau=\emptyset$).
The following ``detour theorem'' provides a weaker form of converse. We will make use of it in one of our proofs.

\begin{theorem}\label{thm:detour}
For all pointed Kripke models $(M,w)$ and $(N,v)$ over the same propositional signature  $\tau$, the following
are equivalent:
\begin{enumerate}
\item $(M,w)$ and $(N,v)$ are modally indistinguishable w.r.t.~$\tau$,
\item There exist pointed Kripke models $(M',w)$ and $(N',v)$, where $(M',w)$ is an ultrapower of $(M,w)$ and $(N',v)$ is an ultrapower of $(N,v)$, such that $(M',w)\bisim_\tau (N',v)$.
\end{enumerate}
\end{theorem}

\begin{figure}
\begin{center}
\begin{tikzcd}[row sep = 32mm, column sep = 40mm]
(M',w) \arrow[r, dash, "\text{$\tau$-bisimilar}", dashed] & (N',v) \\
(M,w) \arrow[u, "\text{ultrapower}", "\text{(\FO-indistinguishable)}"', sloped, dashed] \arrow[r, dash, "\text{modally indistinguishable w.r.t.~$\tau$}" ]
&  (N,v) \arrow[u, "\text{ultrapower}", "(\text{\FO-indistinguishable})"', sloped, dashed]
\end{tikzcd}
\end{center}
\caption{Diagrammatic depiction of Theorem~\ref{thm:detour}.}
\label{fig:detour}
\end{figure}
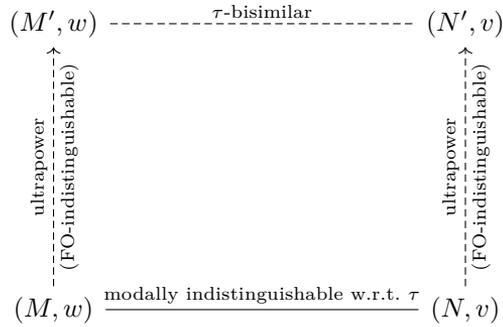

Ultrapowers are a well-known construction in 
model theory (cf.~\cite{hodges1992model}). The precise
definition of ultrapowers is not relevant here. The only property that will
be important for us is the following:
if $B$ is an ultrapower of a structure $A$,
then $A$ and $B$ are indistinguishable for
first-order logic. More precisely,  if a 
pointed Kripke structure
$(M',w)$ is an ultrapower of a pointed Kripke structure $(M,w)$, then
$(M',w)$ and $(M,w)$ are indistinguishable 
for first-order formulas $\phi(x)$. In 
particular, this also means that $(M',w)$
and $(M,w)$ are modally indistinguishable. 
Theorem~\ref{thm:detour} shows that,
although modal indistinguishability does not
imply bisimulation, it implies the existence
of bisimilar ultrapowers (cf.~Figure~\ref{fig:detour}).
We refer to~\cite{BlackburnDeRijkeVenema} for more details
and for a proof of Theorem~\ref{thm:detour}.

\section{First proof (model-theoretic)}
\label{sec:model}

Model-theoretic proofs of Craig interpolation 
typically proceed by contraposition:
from the assumption that there is no Craig
    interpolant for $\phi\to\psi$, we derive that $\phi\to\psi$ is not valid by constructing
    a model for $\phi\land\neg\psi$. 
    Gabbay's original proof of Craig interpolation for modal logic~\cite{Gabbay1972:craig} falls in this category.
    The construction of the model for $\phi\land\neg\psi$ typically involves a method for amalgamating a suitable pair of pointed Kripke models $M,w\models\phi$ and $N,v\models\neg\psi$ into a single pointed Kripke model satisfying $\phi\land\neg\psi$. 
    The \emph{bisimulation product} construction (also known as \emph{zigzag product}), originally introduced in \cite{marx1995algebraic} (see also~\cite{Marx1997:multi}), is such an amalgamation method. 
     Relationships between interpolation properties and amalgamation properties have also been studied from the point of view of algebraic semantics, cf.~the long line of work by Maksimova~\cite{Maks1977, Maks1979, Maksimova1991:amalgamation}. See also~\cite{hooglandthesis, Ven07} and Section~\ref{sec:algebra} of this chapter.

    Before we introduce bisimulation products, we 
    discuss a model-theoretic criterion
    (partly inspired by \cite{BarwiseVanBenthem1999})
    characterizing the existence of an interpolant
    for a given implication. 
        Given modal formulas $\phi,\psi$ and a propositional signature $\tau$, we will say that \emph{$\phi$ entails $\psi$ along $\tau$-bisimulations} if the following holds:
        for every pair of pointed Kripke structures
        $(M,w)$ and $(N,v)$, if $M,w\models\phi$
        and $(M,w)\bisim_\tau (N,v)$, then 
        $N,v\models\psi$. We can think of 
        this as a strong form of entailment: 
        it says when $M,w\models\phi$, this implies that $\psi$
        is true not only in $(M,w)$ but in every pointed Kripke model that
        is $\tau$-bisimilar to $(M,w)$.

    \begin{proposition}[Interpolant Existence Criterion] \label{prop:interpolant-criterion}
        For all modal formulas $\phi$ and $\psi$, 
        the following are equivalent:
        \begin{enumerate}
            \item $\phi\to\psi$ is a valid implication that has a Craig interpolant.
            \item $\phi$ entails $\psi$ along $(\sig(\phi)\cap\sig(\psi))$-bisimulations
        \end{enumerate}
    \end{proposition}

    \begin{proof}
            The direction from 1 to 2 is easy:
            let $\vartheta$ be the interpolant. Then $M,w\models\phi$ implies
            $M,w\models\vartheta$, which
            implies $N,v\models\vartheta$ 
            (by Theorem~\ref{thm:bisim-modal-equiv}),
            which implies $N,v\models\psi$.
            For the direction from 2 to 1, we
            reason by contraposition. First,
            suppose that $\phi\to\psi$ is not 
            valid, i.e., there is a pointed Kripke stucture $M,w\models\phi\land\neg\psi$.
            Then, clearly,
             $\phi$ does not entail $\psi$ along $(\sig(\phi)\cap\sig(\psi))$-bisimulations (because the identity relation is a bisimulation from $(M,w)$ to itself). Next, suppose that 
             $\phi\to\psi$ is valid but does not have a Craig interpolant.

    \medskip\par\noindent
    \emph{Claim 1:} There are pointed Kripke models $M,w\models\phi$ and $N,v\models\neg\psi$ such that $(M,w)$ and $(N,v)$ are modally indistinguishable w.r.t. $\sig(\phi)\cap\sig(\psi)$. 

    \medskip\par\noindent
    \emph{Proof of claim:} a standard argument involving a double application of the \emph{compactness theorem}. 
    The compactness theorem for modal logic 
    states that if a set $\Sigma$ of modal formulas
    is unsatisfiable, then some 
    finite subset of $\Sigma$ is already unsatisfiable.
    Let $\Sigma$ be the set of all 
    modal formulas in signature $\sig(\phi)\cap\sig(\psi)$
    that are entailed by $\phi$. Clearly, 
    every finite subset of $\Sigma\cup\{\neg\psi\}$ is satisfiable
    (otherwise the conjunction of the respective formulas from $\Sigma$ would be a Craig interpolant). Therefore,
    by compactness, there is a pointed Kripke model $N,v\models\Sigma\cup\{\neg\psi\}$.
    Next, let $\Gamma$ be the set of all  modal formulas in signature $\sig(\phi)\cap\sig(\psi)$ true in $N,v$. Clearly,
    every finite subset of $\Gamma\cup\{\phi\}$ is satisfiable (otherwise
    the negation of the conjunction  of the respective formulas from $\Gamma$ 
    would have been a formula belonging to $\Sigma$, contradicting the fact
    that $N,v\models\Sigma$). It follows again by compactness that there
    is a pointed Kripke model 
    $M,w\models\Gamma\cup\{\phi\}$. By construction, we have that
    $M,w\models\phi$ and $N,v\models\neg\psi$ and the two pointed Kripke models
    satisfy the same modal formulas in the signature $\sig(\phi)\cap\sig(\psi)$.

    \medskip\par\noindent
    Next, Theorem~\ref{thm:detour} allows us to ``upgrade'' $(M,w)$ and $(N,v)$ to FO-equivalent
    pointed Kripke models $(M',w)$ and $(N',v)$ such that $M',w\models\phi$ and $N',v\models\neg\psi$ that are bisimilar w.r.t. $\sig(\phi)\cap\sig(\psi)$. Therefore, 
    $\phi$ does not entail $\psi$ along $(\sig(\phi)\cap\sig(\psi))$-bisimulations.
    \end{proof}

    Proposition~\ref{prop:interpolant-criterion} holds not only for the modal logic \logic{K}, but 
    for \emph{every} normal modal logic defined by first-order frame conditions (by the same proof, since ultrapowers preserve all FO sentences, and the modal logic of a first-order frame class is always compact). Indeed, without going into details, it is worth mentioning that Proposition~\ref{prop:interpolant-criterion} holds true 
    for every \emph{canonical} modal logic~\cite{marx1995algebraic}.%
    \footnote{Remarkably, this is the case regardless of 
    whether the logic has the Craig interpolation property. Indeed, for certain modal logics lacking the Craig interpolation property, 
    Proposition~\ref{prop:interpolant-criterion} can be used to decide whether a Craig interpolant exists
    for a given valid implication, 
    cf.~\refchapter{chapter:separation}.}    
    In light of Proposition~\ref{prop:interpolant-criterion}, 
    in order to prove the Craig interpolation 
    property, it suffices to prove that, whenever
    a modal implication $\phi\to\psi$ is valid,
    then $\phi$ entails $\psi$ along $(\sig(\phi)\cap\sig(\psi))$-bisimulations.
    This is what we will use bisimulation products
    for.


\begin{definition}[Bisimulation Product]
Let $M=(W,R,V)$ and $M'=(W',R',V')$ be Kripke models over the same propositional signature.
\begin{itemize}
    \item The \emph{direct product} $M\times N$ 
    is the Kripke model whose domain is the cartesian product $W\times W'$; whose
    accessibility relation 
    consists of all pairs of pairs $\pair{(w_1,u_1),(w_2,u_2)}$ such that $(w_1,w_2)\in R \text{ and } (u_1, u_2)\in R'$; and 
    where a pair $(w,u)$ belongs to the valuation of a proposition letter $p$ iff $w\in V(p)$ and $u\in V'(p)$.
    \item A \emph{subdirect product} of $M$ and $N$
    is any induced substructure of $M\times N$.%
    \footnote{By an \emph{induced substructure} of a Kripke model $M=(W,R,V)$ we mean a Kripke model $M'=(W',R',V')$ where $W'\subseteq W$ and where $R'$ and $V'$ are obtained by restricting $R$ and $V$ to $W'$.}
    The subdirect product of $M$ and $N$ with domain $Z\subseteq W\times W'$ is denoted by  $M\times_Z N$.
    \item A \emph{bisimulation product} of $M$ and $N$ is a subdirect product $M\times_Z N$ such that $Z$ is a bisimulation between $M$ and $N$. If $Z$ is a full bisimulation, we also call $M\times_Z N$ a \emph{full bisimulation product}.
\end{itemize}
The same definitions apply also to Kripke frames (without the valuations).%
\footnote{For Kripke frames, (full) bisimulation products can be equivalently defined as subdirect products whose natural projections are (surjective) bounded morphisms.}
\end{definition}
Recall that a \emph{full bisimulation} is a bisimulation whose domain is the full domain of $A$ and whose range is the full domain of $B$. 
The fullness requirement will in fact not play a major role in what follows.
However, we should point out that the most common definition of a bisimulation product in prior literature~\cite{marx1995algebraic,Marx1997:multi} is in fact what we call full bisimulation products here.

\begin{example}\label{ex:bisprod}
Consider the following two Kripke models in the signature $\{q\}$. 
    \[  M:
\begin{tikzcd}[row sep = -1mm]
                    &                                  & w^1_3 ~^q\\
~^q~w_1 \arrow[r,""] & w_2 \arrow[ru,""] \arrow[rd,""]  &   \\
                    &                                  & w^2_3 ~^q 
\end{tikzcd}
~~~~~~~~~ N:
\begin{tikzcd}[row sep = -1mm]
                                  & v^1_2  \arrow[rd,""] & \\
^q~v_1 \arrow[ru,""]\arrow[rd,""] &                      & v_3 ~^q \\
                                  & v^2_2  \arrow[ru,""] &  
\end{tikzcd}
\]
The binary relation $Z=\{\pair{w_1,v_1}, \pair{w_2,v_2^1}, \pair{w_2,v_2^2}, \pair{w_3^1,v_3}, \pair{w_3^2,v_3}\}$ is a bisimulation (indeed, a full bisimulation), and the bisimulation
product $M\times_{Z} N$ is the following Kripke model:
\[
\begin{tikzcd}[row sep = -1mm]
                           & \pair{w_2,v^1_2} \arrow[r,""] 
                             \arrow[rdd,""]                   & \pair{w^1_3,v_3}~^q \\
^q~\pair{w_1,v_1} 
\arrow[ru,""]\arrow[rd,""] &                                   &   \\
                           & \pair{w_2,v^2_2} 
                              \arrow[r,""] \arrow[ruu,""]      & \pair{w^2_3,v_3}~^q 
\end{tikzcd}
\]
\eoe\end{example}

As promised, bisimulation products provide an
amalgamation method. The precise amalgamation
property we need is the following:

\begin{theorem}[\cite{marx1995algebraic}]\label{thm:amalgamation}
     Whenever $(M_1,w_1)\bisim_{\sigma\cap\tau} (M_2,w_2)$, there is  a
       pointed Kripke model $(N,v)$ such that
       $(N,v)\bisim_\sigma (M_1,w_1)$ and $(N,v)\bisim_\tau (M_2,w_2)$.
\end{theorem}

\begin{proof}
   Let $M_1=(W_1,R_1,V_1)$ and $M_2=(W_2,R_2,V_2)$, and
   let $Z$ be a $\sigma\cap\tau$-bisimulation between $M_1$ and $M_2$
   such that $\pair{w_1,w_2}\in Z$.
   %
   %
   Let $M_1\times_Z M_2$ be the bisimulation product in signature $\sigma\cap\tau$.
   Finally, we expand $M_1\times_Z M_2$ to a Kripke model $N=(W,R,V)$ in the full signature $\sigma\cup\tau$ by extending the valuation function
   $V$ of $M_1\times_Z M_2$ such that
   $V(p)=\{\pair{u,v}\in Z\mid u\in V_1(p)\}$ for $p\in \sigma\setminus\tau$, and setting
      $V(p)=\{\pair{u,v}\in Z\mid v\in V_2(p)\}$ for $p\in\tau\setminus\sigma$. 
It can be verified that the relations $Z_1=\{\pair{\pair{u,v},u}\mid \pair{u,v}\in Z\}$ and $Z_2=\{\pair{\pair{u,v},v}\mid \pair{u,v}\in Z\}$ are bisimulations 
from the structure thus obtained to $M_1$ and $M_2$ respectively, with respect to the signature $\sigma$ and $\tau$, respectively. 
\end{proof}

\begin{example}
This example builds on Example~\ref{ex:bisprod}.
Consider the following two Kripke models in the signature $\sigma=\{p,q\}$ and $\tau=\{q,r\}$, respectively:
    \[ M:
\begin{tikzcd}[row sep = -1mm]
                    &                                  & ~~ w^1_3 ~^{q,p}\\
~^q~w_1 \arrow[r,""] & w_2 \arrow[ru,""] \arrow[rd,""]  &   \\
                    &                                  & w^2_3  ~^q 
\end{tikzcd}
~~~~~~~~~ N:
\begin{tikzcd}[row sep = -1mm]
                                  & v^1_2  \arrow[rd,""] & \\
~^q~v_1 \arrow[ru,""]\arrow[rd,""] &                      & v_3 ~^q \\
                                  & v^2_2 \arrow[ru,""] &  \\
                                  & ~~~~~~^r
\end{tikzcd}
\]
Observe that the reducts of $M$ and $N$ to the common signature $\sigma\cap\tau$
are precisely the Kripke models in Example~\ref{ex:bisprod}. As we saw in 
Example~\ref{ex:bisprod},
the binary relation \[Z=\{\pair{w_1,v_1}, \pair{w_2,v_2^1}, \pair{w_2,v_2^2}, \pair{w_3^1,v_3}, \pair{w_3^2,v_3}\}\] is a bisimulation with respect to the 
common signature $\sigma\cap\tau$. Following the proof of Theorem~\ref{thm:amalgamation}, to amalgamate the two structures, we 
take the bisimulation product with respect to the common signature, and 
we expand it to the full signature $\sigma\cup\tau$ by interpreting
the proposition letters in $\sigma\setminus\tau$ according the left projection 
and interpreting the proposition letters in $\tau\setminus\sigma$ according
to the right projection. In this way, we obtain:
\[
\begin{tikzcd}[row sep = -1mm]
                           & \pair{w_2,v^1_2} \arrow[r,""] 
                             \arrow[rdd,""]                   & ~~ \pair{w^1_3,v_3}~^{q,p} \\
~^q~\pair{w_1,v_1} 
\arrow[ru,""]\arrow[rd,""] &                                   &   \\
                           & \pair{w_2,v^2_2} 
                              \arrow[r,""] \arrow[ruu,""]      & \pair{w^2_3,v_3}~^q  \\
                          & ~~~~~~~~~~~^r
\end{tikzcd}
\]
It can be verified that this amalgamated structure is bisimilar to $M$
w.r.t.~$\sigma$ and bisimilar to $N$ w.r.t.~$\tau$.
\eoe\end{example}

\begin{mainproof}
    Let $\phi, \psi$ be modal formulas such that $\phi\to\psi$ does \emph{not} have a Craig interpolant. Then, by Proposition~\ref{prop:interpolant-criterion}, there are pointed Kripke models
    $M,w\models\phi$ and $N,v\models\neg\psi$ that are bisimilar w.r.t. $\sig(\phi)\cap\sig(\psi)$. By Theorem~\ref{thm:amalgamation}, there is a pointed Kripke model that is bisimilar in signature $\sig(\phi)$ to $(M,w)$ and bisimilar in 
    signature $\sig(\psi)$ to $(N,v)$. It immediately follows that this 
    pointed Kripke structure satisfies $\phi\land\neg\psi$ and therefore
    $\phi\to\psi$ is not a valid modal implication.
\end{mainproof}

This 
    approach generalizes to the modal logic of any frame class defined by first-order frame conditions that are preserved by bisimulation products, or, alternatively, by full
    bisimulation products and generated subframes (since a bisimulation can always be ``upgraded'' to a full bisimulation by taking generated submodels). This
    includes all frame classes defined by universal Horn sentences, since these are closed under subdirect
    products and bisimulations are an instance of subdirect products. See also~Section~\ref{sec:algebra}.
    In
    \cite{tenCate2005}, a precise syntactic characterization is given of the first-order frame conditions that are preserved by bisimulation products (which forms a proper generalization of universal Horn sentences). In order to state this characterization, let a \emph{strict p-formula} be a first-order formula obtained from atomic formulas (including equality statements, $\top$ and $\bot$) using conjunction, disjunction, existential quantifiers, and bounded universal quantifiers of the 
    form $\forall y(Rxy\to\cdot)$, where $x$ and $y$ are distinct variables. 

    \begin{theorem}[\cite{tenCate2005}]
    A first-order sentence $\phi$ is preserved by bisimulation products iff $\phi$
is equivalent to a conjunction of sentences of the form
$\forall\textbf{x}(\psi\to\chi)$, where $\psi$ is a strict p-formula and $\chi$ is
 an atomic formula.
\end{theorem}

    Many interesting first-order frame conditions are covered by this. For instance, all universal Horn conditions. A notable exception is the Church--Rosser condition
    $\forall xyz(Rxy\land Ryz\to \exists u(Ryu\land Rzu))$. Indeed, the modal logic of the frame class defined by this condition lacks Craig interpolation (cf.~\cite{tenCate2005}). The modal logic \logic{K.Alt_2} (i.e., the logic of the class of frames in which every world has at most two successors) is an example of a modal logic that is not closed under bisimulation products but still has Craig interpolation (cf.~\cite{DBLP:conf/stacs/JungK25}).
    
In summary:

\pro{The model theoretic proof establishes Craig interpolation for the modal logics of many first-order definable frame classes.}


\con{The model theoretic proof does not apply to non-elementary logics such as $\mathsf{GL}$ and $\mathsf{Grz}$, nor to non-elementary language extensions such as the modal $\mu$-calculus.}

\con{This proof technique does not extend in any obvious way to prove uniform interpolation.}

\con{The model theoretic proof does not yield an effective method for constructing Craig interpolants.}

\section{Second proof (proof-theoretic)}

In this section we turn to a proof-theoretic proof of Craig interpolation for \K. The proof is called proof-theoretic because it makes use of a proof system for \K, in this case a sequent calculus, which is introduced below. 
Proof-theoretic proofs usually have nice benefits:\footnote{We do not mean to imply that these benefits only hold for proof-theoretic proofs, other methods of proof may have similar or other benefits.}
most of them, and this certainly holds for the proof in this section, are constructive in that they contain a method to explicitly construct an interpolant from the proof of the implication that is interpolated. Another benefit is that they easily generalize to proof-theoretic proofs of CIP for other logics with a similar proof system. Also, strengthenings of the interpolation property are usually easy to obtain from proof-theoretic proofs, we will see examples of this in Section~\ref{sec:harvest} and in \refchapter{chapter:prooftheory}. 

The proof system of choice in this section is a {\it sequent calculus} (\GthK), a proof system that derives sequents instead of formulas. {\it Sequents} are expressions of the form $(\Ga \seq \De)$, where $\Ga$ and $\De$ are finite multisets of modal formulas. Their intended interpretation is $I(\Ga \seq \De) \defn (\ben\Ga \imp \bof\De)$. 
Because of the choice for sequents rather than formulas, rules for many logical operators (connectives, modalities) can be expressed elegantly, and the sequent calculus \GthK\ in this section is an example of this phenomenon. To prove that \K\ has interpolation, the notion is first generalized to sequents in Section~\ref{subsec:partitions} under the name sequent interpolation. Then it is shown in Section~\ref{subsec:main} that the calculus \GthK\ has sequent interpolation which, as we will see, implies that \K\ has interpolation. This method of proof is called the {\it Maehara method}, named after its inventor. A proof for classical logic that uses that method can be found in \cite{takeuti87}.

We prove that \GthK\ has sequent interpolation by showning that it has rule interpolation, which roughly means the following. The proof that a derivable sequent has a sequent interpolant proceeds by induction on the depth of the derivation of the sequent: the induction hypothesis assumes that the premise(s) of the last inference have interpolants and then it is shown that so has the conclusion. Rule interpolation means that for each rule separately the interpolant for the conclusion can be obtained from those of the premises, independent of the rest of the proof, it can be obtained {\it locally}. The precise formulation can be found in Section~\ref{subsec:partitions}. Usually, in proofs of CIP that use sequent calculi, this locality of the definition of interpolants is not explicitly mentioned, although in most cases the local version is what is actually proven.

\subsection{Preliminaries}
 \label{sec:preliminaries}

\paragraph*{Sequent calculi}
We denote finite multisets of formulas by $\Ga,\Pi,\De,\Sig$. $\Ga \cup \Pi$, also denoted $\Ga,\Pi$, is the multiset that contains all formulas of $\Ga$ and $\Pi$ (as multisets). 
A {\it sequent} is an expression $(\Ga \seq \De)$, with the intended interpretation $I(\Ga \seq \De) \defn (\ben\Ga \imp \bof\De)$, where $\ben\varnothing \defn \top$ and $\bof\varnothing \defn \bot$. 
A {\it (sequent) rule}  
is a pair consisting of a sequence of sequents (the premises) and one sequent (the conclusion). It is an {\it (sequent) axiom} if the set of premises is empty. A {\it sequent calculus} \G\ consists of a set of axioms and rules. Thus a sequent calculus is like a usual proof system but about sequents instead of formulas. 

As in the introduction of this section, the signature of a multiset $\Ga$ consists of all propositional letters that occur in some formula in $\Ga$ and is denoted $\sgn(\Ga)$. For sequents, $\sgn(\Ga \seq \De) \defn \sgn(\Ga)\cup \sgn(\De)$. The {\it common language} of two sequents $S_1$  and $S_2$ is 
$\sgn(S_1) \cap \sgn (S_2)$. A positive (negative) occurence of a prositional letter in a formula is defined as usual. A prositional letter occurs {\it positively (negatively)} in a sequent $S$ whenever it occurs positively (negatively) in $I(S)$.

\paragraph*{The calculus \GthK}
 \label{sec:calculiu}
Calculus $\GthK$ consists of the calculus $\Gthp$ (Figure~\ref{fig:gth}) plus the modal rule $R_{\fnbx}$ (Figure~\ref{fig:rulek}). $\Gthp$ is the standard calculus without structural rules for $\CPC$ from \cite{troelstra&schwichtenberg96}, which is a restriction to the propositional language of the calculus \Gth\ that already occurs in \cite[Chapter XV]{kleene52} under that same name. 
For simplicity we leave $\dm$ out of the language and treat it as an abbreviation of $\neg\bx\neg$, so that the only modal rules we need are for the $\bx$-operator. When convenient, axioms are sometimes called zero-premise rules. 
\begin{figure}[t]
\begin{center}\small
\dotfill
 \\ \
 \\
 $\begin{array}{lll}
 \multicolumn{3}{l}{
 \AxiomC{} 
 \RightLabel{Ax$_{\text{id}}$} 
 \UnaryInfC{$\Ga,p \seq p,\De$}
 \DisplayProof \ \ \ \ 
 \AxiomC{} 
 \RightLabel{Ax$_{\bot}$} 
 \UnaryInfC{$\Ga,\bot\seq \De$}
 \DisplayProof  \ \ \ \ 
 \AxiomC{} 
 \RightLabel{Ax$_{\top}$} 
 \UnaryInfC{$\Ga\seq \top,\De$}
 \DisplayProof } \\
 \\ 
 \AxiomC{$\Ga,\phi \seq \De$}
 \LeftLabel{$\neg_r$}
 \UnaryInfC{$\Ga \seq \neg\phi,\De$}
 \DisplayProof & \hspace{.7cm} & 
 \AxiomC{$\Ga\seq \phi,\De$}
 \LeftLabel{$\neg_l$}
 \UnaryInfC{$\Ga,\neg\phi \seq \De$}
 \DisplayProof \\
 \\
 \AxiomC{$\Ga\seq \phi,\De$}
 \AxiomC{$\Ga\seq \psi,\De$}
 \LeftLabel{$\en_r$} 
 \BinaryInfC{$\Ga \seq \phi \en \psi,\De$}
 \DisplayProof & & 
 \AxiomC{$\Ga, \phi, \psi \seq \De$}
 \LeftLabel{$\en_l$} 
 \UnaryInfC{$\Ga, \phi\en \psi \seq \De$}
 \DisplayProof \\
 \\
 \AxiomC{$\Ga \seq \phi,\psi,\De$}
 \LeftLabel{$\of_r$}
 \UnaryInfC{$\Ga \seq \phi\of\psi,\De$}
 \DisplayProof & & 
 \AxiomC{$\Ga, \phi \seq \De $}
 \AxiomC{$\Ga, \psi \seq \De $}
 \LeftLabel{$\of_l$}
 \BinaryInfC{$\Ga, \phi \of \psi \seq \De $}
 \DisplayProof 
 \end{array}$ 
 \\ \ 
 \\
 \dotfill
\end{center}
\vspace{-.3cm}
\caption{Sequent calculus $\Gthp$}
 \label{fig:gth}
\end{figure}

\begin{figure}[t]
\begin{center}\small 
 $\AxiomC{$\Ga \seq \phi$}
 \LeftLabel{$R_{\fnbx}$} 
 \UnaryInfC{$\Pi,\bx\Ga \seq \bx\phi,\De$}
 \DisplayProof$ 
\end{center}
\vspace{-.3cm}
\caption{Sequent calculus $\GthK$: \Gthp\ plus the rule $R_{\fnbx}$.}
 \label{fig:rulek}
\end{figure}

A {\it derivation} in \GthK\ of a sequent $S$ is a finite tree labelled with sequents such that the root is labelled with $S$ and for any node with label $S'$ that is not a leaf there is a(n instance of a) rule in \GthK\  such that its conclusion is $S'$ and its premises are exactly the direct successors of $S'$ in the tree. The derivation is a {\it proof} once all leafs are (instances of) axioms. $S$ {\it has a proof} or {\it is provable} in \GthK, denoted $\af S$, if there is a proof in \GthK\ that ends with $S$. The {\it depth} of a proof is the length of its longest branch. Thus axioms are proofs of depth one.  

The {\it principal}, {\it active}, and {\it auxiliary formulas} of a rule are defined as usual for the rules in \Gthp: the displayed formula in the conclusion is {\it principal}, and the displayed formulas in the premises are the {\it active} formulas. The formulas in $\Ga$ and $\De$ are {\it auxiliary}. In axioms Ax$_{\bot}$ and Ax$_{\top}$ formulas $\bot$ and $\top$ are {\it principal}, respectively, and in axiom Ax$_{\rm id}$ both $p$ are {\it principal}.  
For $R_{\fnbx}$ the principal formula is $\bx\phi$ and all formulas in $\bx\Ga$ in the conclusion and all formulas in $\Ga$ as wel as $\phi$ in the premise are active. The auxiliary formulas are the formulas in $\Pi$ and $\De$.  

The following theorem establishes that proof system \GthK\ is sound and complete with respect to logic \K, a proof can be found in e.g.~\cite{poggiolesi11}.  

\begin{theorem} 
For any sequent $S$: $S$ is provable in \GthK\ if and only if $I(S)$ holds in \K. 
\end{theorem}

The following well-known and easy-to-prove fact about \GthK, which we need below, is its closure under weakening. For a proof, see \cite{troelstra&schwichtenberg96}. 

\begin{lemma}({\it Weakening lemma})\\
 \label{lem:weak}
If $\af \Ga\seq \De$, then $\af \Ga,\phi \seq \De$ and $\af \Ga \seq \phi,\De$.
\end{lemma}

\subsection{Sequent interpolation}
 \label{sec:seqip}
As explained at the beginning of this section, there is a generalization of Craig interpolation to sequents and sequent calculi such that whenever a logic has a sequent calculus that satisfies this sequent interpolation, the logic has Craig interpolation. That \K\ has interpolation is shown by proving that \GthK\ has sequent interpolant, as is done in Section~\ref{subsec:main}.

\subsection{Split sequents}
 \label{subsec:partitions}
For the generalization of interpolation to sequents, 
we need the notion of split: a {\it split} of a sequent $\Ga \seq \De$ is an expression $\Ga_1;\Ga_2 \seq \De_1;\De_2$ (called a {\it split sequent}) such that $\Ga=\Ga_1\cup\Ga_2$ and $\De=\De_1\cup\De_2$ and $\Ga_1;\Ga_2$ and $\De_1;\De_2$ are ordered pairs, so that $\Ga_1;\Ga_2$ is in general not equal to $\Ga_2;\Ga_1$. The formulas in $\Ga_1,\De_1$ belong to the {\it left split} and the formulas in $\Ga_2,\De_2$ belong to the {\it right split}. Let $\parseq$ denote the set of all split sequents that consist of formulas in our modal language. We denote split sequents with $S^p$ or $S_i^p$. 

Using the terminology from \cite{gherardietal2020}, a formula $\varchi$ is a {\it split-interpolant} for split $\Ga_1;\Ga_2 \seq \De_1;\De_2$ if $\af \Ga_1 \seq \varchi,\De_1$ and $\af \Ga_2,\varchi \seq \De_2$ and $\varchi$ is in the common language of $\Ga_1\seq\De_1$ and $\Ga_2\seq\De_2$, i.e.\ $\sgn(\varchi) \subseteq \sgn(\Ga_1\seq\De_1)\cap\sgn(\Ga_2\seq\De_2)$. We say {\it interpolant} instead of {\it split-interpolant} when the context determines which interpolant we mean, which is most of the time. A sequent calculus has {\it sequent interpolation} if any split of any derivable sequent has an interpolant. 

\subsection{Rule interpolation}
The proof that \K\ has CIP follows from a stronger statement, Lemma~\ref{lem:ruleip}, that states that every rule in \GthK\ interpolates, which is defined as follows. Here axioms are considered to be zero-premise rules. Let $\form$ denote the set of all formulas in our modal language. A one-premise (two- and zero-premise) rule $R$ {\it interpolates} if there exists a function $\delta_R:\form \times \parseq\imp \form$ ($\delta_R:\form\times\form \times \parseq\imp \form$ and $\delta_R:\parseq\imp \form$ in the two-premise and zero-premise case resp.) such that for every instance of the rule with conclusion $S$ and premise(s) $S_1$ (and $S_2$), for every partition $S^p$ of $S$, there exist splits $S_1^p$ (and $S_2^p$) of the premises such that for any interpolant(s) $\varchi_1$ of $S_1^p$ (and $\varchi_2$ of $S_2^p$) formula $\delta_R(\varchi_1,S^p)$ \big($\delta_R(\varchi_1,\varchi_2,S^p)$ and $\delta_R(S^p)$ in the two-premise and zero-premise case resp.\big) is an interpolant of $S^p$. Such a function $\delta_R$ is an {\it interpolant function} for $R$. Note that the domain of the delta-function for axioms is $\parseq$, so that in this case the interpolant of the conclusion only depends on the split of the conclusion, as it should be since in the case of axioms there are no premises. 

\subsection{Rule interpolation for \GthK}
 \label{sec:ipfunction}
In Figures~\ref{fig:ruleipaxioms} and ~\ref{fig:ruleiprules} we define for every rule of \GthK\ a function $\delta_R$ that we prove to be an interpolation function of $R$ in Lemma~\ref{lem:ruleip}, in the following way. 
For the axioms we distinguish cases according to the principal formula(s) (in blue) being in the left or the right side of the split. Note that for axiom Ax$_{\rm id}$ there are four cases because it has two principal formulas. 
For a given split $S^p$ of axiom $R$, consider in Figure~\ref{fig:ruleipaxioms} that particular split of $R$ and let $\delta_R(S^p)$ be the formula (in red) on top of that sequent arrow. For example, 
\begin{center}\small
 $\delta_{\text{Ax}_{\text{id}}}(\Ga_1;\Ga_2,p \seq p,\De_1;\De_2) = \neg p.$
\end{center}
For the rules, Figure~\ref{fig:ruleiprules}, we distinguish cases according to the presence of the principal formula (in blue) in the left or the right side of the split. As one can see, given such a split of the conclusion, the split that is chosen for the premise(s) is in all cases the one in which the active formula(s) is (are) in the same part of the split as the principal formula, and all other formulas remain as they are. The formula(s) on top of the sequent arrow(s) in the premise(s) is (are) the assumed interpolant(s) $\varchi_1$ (and $\varchi_2$) of the split premises and $\delta_R(\varchi_1,S^p)$ \big($\delta_R(\varchi_1,\varchi_2,S^p)$\big) is defined to be the formula (in red) on top of the sequent arrow in the conclusion. For example, 
\begin{center}\small
 $\delta_{R_{\bx}}\big(\varchi,(\Pi_1,\bx\Ga_1;\Pi_2,\bx \Ga_2 \seq \De_1;\De_2,\bx\phi)\big) = \bx\varchi.$
\end{center}

\begin{figure}[t]
\begin{center}\small
 \dotfill
 \\ \ 
 \\
$\begin{array}{lll}
 \AxiomC{}
 \LeftLabel{\tiny{$LL$(Ax$_{\text{id}})$}}
 \UnaryInfC{$\boldsymbol{p},\Ga_1; \Ga_2 \stackrel{\raisebox{2pt}{\textbf{$\bot$}}}{\seq} \boldsymbol{p},\De_1;\De_2$}
 \DisplayProof & \hspace{.7cm} & 
 \AxiomC{}
 \LeftLabel{\tiny{$LR$(Ax$_{\text{id}})$}}
 \UnaryInfC{$\boldsymbol{p},\Ga_1; \Ga_2 \stackrel{\raisebox{2pt}{$p$}}{\seq} \De_1;\De_2,\boldsymbol{p}$}
 \DisplayProof \\
 \\
 \AxiomC{}
 \LeftLabel{\tiny{$RL$(Ax$_{\text{id}})$}}
 \UnaryInfC{$\Ga_1; \Ga_2,\boldsymbol{p} \stackrel{\raisebox{2pt}{$\neg p$}}{\seq} \boldsymbol{p},\De_1;\De_2$}
 \DisplayProof & \hspace{.7cm} & 
 \AxiomC{}
 \LeftLabel{\tiny{$RR$(Ax$_{\text{id}})$}}
 \UnaryInfC{$\Ga_1; \Ga_2,\boldsymbol{p} \stackrel{\raisebox{2pt}{$\top$}}{\seq} \De_1;\De_2,\boldsymbol{p},$}
 \DisplayProof \\
 \\
 \AxiomC{}
 \LeftLabel{\tiny{$L$(Ax$_{\bot})$}}
 \UnaryInfC{${\bf \bot},\Ga_1; \Ga_2 \stackrel{\raisebox{2pt}{$\bot$}}{\seq} \De_1;\De_2$}
 \DisplayProof & \hspace{.7cm} & 
 \AxiomC{}
 \LeftLabel{\tiny{$R$(Ax$_{\bot})$}}
 \UnaryInfC{$\Ga_1; \Ga_2,{\bf \bot} \stackrel{\raisebox{2pt}{$\top$}}{\seq} \De_1;\De_2$}
 \DisplayProof \\
 \\
 \AxiomC{}
 \LeftLabel{\tiny{$L$(Ax$_{\top})$}}
 \UnaryInfC{$\Ga_1; \Ga_2 \stackrel{\raisebox{2pt}{$\bot$}}{\seq} {\bf \top},\De_1;\De_2$}
 \DisplayProof & \hspace{.7cm} & 
 \AxiomC{}
 \LeftLabel{\tiny{$R$(Ax$_{\top})$}}
 \UnaryInfC{$\Ga_1; \Ga_2 \stackrel{\raisebox{2pt}{$\top$}}{\seq} \De_1;\De_2,{\bf \top}$}
 \DisplayProof 
 \end{array}$
 \\ \ 
 \\
 \dotfill
\end{center}
 \vspace{-.3cm}
 \caption{Rule interpolation for the axioms of \GthK}
   \label{fig:ruleipaxioms}
 \end{figure}

\begin{figure}[t]
\begin{center}\small
 \dotfill 
 \\ \ 
 \\
 $\begin{array}{@{}ll@{}}
 \AxiomC{$\Ga_1 ; \Ga_2 \stackrel{\raisebox{2pt}{$\varchi$}}{\seq} \boldsymbol{\phi},\De_1;\De_2$}
 \LeftLabel{{\tiny $L(\neg_l)$}}
 \UnaryInfC{$\boldsymbol{\neg\phi},\Ga_1; \Ga_2 \stackrel{\raisebox{2pt}{$\varchi$}}{\seq} \De_1;\De_2$}
 \DisplayProof 
 & 
 \AxiomC{$\Ga_1 ; \Ga_2\stackrel{\raisebox{2pt}{$\varchi$}}{\seq} \De_1;\De_2,\boldsymbol{\phi} $}
 \LeftLabel{{\tiny $R(\neg_l)$}}
 \UnaryInfC{$\Ga_1; \Ga_2,\boldsymbol{\neg\phi} \stackrel{\raisebox{2pt}{$\varchi$}}{\seq}\De_1;\De_2$}
 \DisplayProof  
 \\
 \\
 \AxiomC{$\boldsymbol{\phi} ,\Ga_1 ; \Ga_2 \stackrel{\raisebox{2pt}{$\varchi$}}{\seq} \De_1;\De_2$}
 \LeftLabel{{\tiny $L(\neg_r)$}}
 \UnaryInfC{$\Ga_1; \Ga_2 \stackrel{\raisebox{2pt}{$\varchi$}}{\seq} \boldsymbol{\neg\phi},\De_1;\De_2$}
 \DisplayProof 
 &  
 \AxiomC{$\Ga_1 ; \Ga_2,\boldsymbol{\phi} \stackrel{\raisebox{2pt}{$\varchi$}}{\seq}\De_1;\De_2$}
 \LeftLabel{{\tiny $R(\neg_r)$}}
 \UnaryInfC{$\Ga_1; \Ga_2 \stackrel{\raisebox{2pt}{$\varchi$}}{\seq}\De_1;\De_2,\boldsymbol{\neg\phi}$}
 \DisplayProof 
 \\
 \\
 \AxiomC{$\boldsymbol{\phi,\psi}, \Ga_1 ; \Ga_2 \stackrel{\raisebox{2pt}{$\varchi$}}{\seq}\De_1;\De_2$}
 \LeftLabel{{\tiny $L(\en_l)$}}
 \UnaryInfC{$\boldsymbol{\phi\en\psi},\Ga_1; \Ga_2 \stackrel{\raisebox{2pt}{$\varchi$}}{\seq}\De_1;\De_2$}
 \DisplayProof 
 & 
 \AxiomC{$\Ga_1 ; \Ga_2,\boldsymbol{\phi,\psi}\stackrel{\raisebox{2pt}{$\varchi$}}{\seq}\De_1;\De_2$}
 \LeftLabel{{\tiny $R(\en_l)$}}
 \UnaryInfC{$\Ga_1; \Ga_2,\boldsymbol{\phi\en\psi} \stackrel{\raisebox{2pt}{$\varchi$}}{\seq}\De_1;\De_2$}
 \DisplayProof 
 \\
 \\
 \AxiomC{$\Ga_1 ; \Ga_2 \stackrel{\raisebox{2pt}{$\varchi$}}{\seq}\boldsymbol{\phi,\psi},\De_1;\De_2$}
 \LeftLabel{{\tiny $L(\of_r)$}}
 \UnaryInfC{$\Ga_1; \Ga_2 \stackrel{\raisebox{2pt}{$\varchi$}}{\seq}\boldsymbol{\phi\of\psi},\De_1;\De_2$}
 \DisplayProof 
 & 
 \AxiomC{$\Ga_1 ; \Ga_2 \stackrel{\raisebox{2pt}{$\varchi$}}{\seq},\boldsymbol{\phi,\psi},\De_1;\De_2$}
 \LeftLabel{{\tiny $R(\of_r)$}}
 \UnaryInfC{$\Ga_1; \Ga_2 \stackrel{\raisebox{2pt}{$\varchi$}}{\seq}\De_1;\De_2,\boldsymbol{\phi\of\psi}$}
 \DisplayProof 
 \\
 \\
 \AxiomC{$\boldsymbol{\phi},\Ga_1 ; \Ga_2 \stackrel{\raisebox{2pt}{$\varchi_1$}}{\seq}\De_1;\De_2$}
 \AxiomC{$\boldsymbol{\psi},\Ga_1 ; \Ga_2 \stackrel{\raisebox{2pt}{$\varchi_2$}}{\seq}\De_1;\De_2$}
 \LeftLabel{{\tiny $L(\of_l)$}}
 \insertBetweenHyps{\hskip -1pt}
 \BinaryInfC{$\boldsymbol{\phi \of \psi},\Ga_1; \Ga_2 \stackrel{\raisebox{2pt}{$\varchi_1\of\varchi_2$}}{\seq}\De_1;\De_2$}
 \DisplayProof 
 & 
 \AxiomC{$\Ga_1 ; \Ga_2,\boldsymbol{\phi}\stackrel{\raisebox{2pt}{$\varchi_1$}}{\seq}\De_1;\De_2$}
 \AxiomC{$\Ga_1 ; \Ga_2,\boldsymbol{\psi} \stackrel{\raisebox{2pt}{$\varchi_2$}}{\seq}\De_1;\De_2$}
 \LeftLabel{{\tiny $R(\of_l)$}}
 \insertBetweenHyps{\hskip -1pt}
 \BinaryInfC{$\Ga_1; \Ga_2,\boldsymbol{\phi\of\psi} \stackrel{\raisebox{2pt}{$\varchi_1 \en\varchi_2$}}{\seq}\De_1;\De_2$}
 \DisplayProof 
 \\
 \\
 \AxiomC{$\Ga_1 ; \Ga_2  \stackrel{\raisebox{2pt}{$\varchi_1$}}{\seq}\boldsymbol{\phi},\De_1;\De_2$}
 \AxiomC{$\Ga_1 ; \Ga_2 \stackrel{\raisebox{2pt}{$\varchi_2$}}{\seq}\boldsymbol{\psi},\De_1;\De_2$}
 \LeftLabel{{\tiny $L(\en_r)$}}
 \insertBetweenHyps{\hskip -1pt}
 \BinaryInfC{$\Ga_1; \Ga_2 \stackrel{\raisebox{2pt}{$\varchi_1 \of \varchi_2$}}{\seq}\boldsymbol{\phi \en\psi},\De_1;\De_2$}
 \DisplayProof 
 & 
 \AxiomC{$\Ga_1 ; \Ga_2 \stackrel{\raisebox{2pt}{$\varchi_1$}}{\seq}\De_1;\De_2,\boldsymbol{\phi}$}
 \AxiomC{$\Ga_1 ; \Ga_2 \stackrel{\raisebox{2pt}{$\varchi_2$}}{\seq}\De_1;\De_2,\boldsymbol{\psi}$}
 \LeftLabel{{\tiny $R(\en_r)$}}
 \insertBetweenHyps{\hskip -1pt}
 \BinaryInfC{$\Ga_1; \Ga_2\stackrel{\raisebox{2pt}{$\varchi_1 \en \varchi_2$}}{\seq}\De_1;\De_2,\boldsymbol{\phi \en \psi}$}
 \DisplayProof 
 \\
 \\
 \AxiomC{$\Ga_1;\Ga_2 \stackrel{\raisebox{2pt}{$\varchi$}}{\seq}\boldsymbol{\phi} ;\ $}
 \LeftLabel{{\tiny $L(R_{\bx})$}}
 \UnaryInfC{$\Pi_1,\bx\Ga_1;\Pi_2,\bx\Ga_2 \stackrel{\raisebox{2pt}{$\neg\bx\neg\varchi$}}{\seq}\boldsymbol{\bx\phi},\De_1;\De_2$}
 \DisplayProof 
 & 
 \AxiomC{$\Ga_1;\Ga_2 \stackrel{\raisebox{2pt}{$\varchi$}}{\seq} \ ;\boldsymbol{\phi} $}
 \LeftLabel{{\tiny $R(R_{\bx})$}}
 \UnaryInfC{$\Pi_1,\bx\Ga_1;\Pi_2,\bx\Ga_2 \stackrel{\raisebox{2pt}{$\bx\varchi$}}{\seq} \De_1; \De_2,\boldsymbol{\bx\phi}$}
 \DisplayProof
\end{array}$ 
\\ \ 
\\
\dotfill 
\end{center}
 \vspace{-.3cm}
\caption{Rule interpolation for the rules of \GthK}
 \label{fig:ruleiprules}
\end{figure}

Note that the definition of the interpolant function in Figures~\ref{fig:ruleipaxioms} and ~\ref{fig:ruleiprules} is schematic in the sense that for the definition of $\delta_R$ we do not have to distinguish different instances of rule $R$, we just distinguish according to the location of the principal formula(s) in the split of the conclusion.

\subsection{Main theorem}
 \label{subsec:main}
 
\begin{lemma} ({\it Rule interpolation})
 \label{lem:ruleip}\\
Every rule $R$ in \GthK\ interpolates, with interpolant function $\delta_R$. 
\end{lemma}
\begin{proof}
Recall that we consider axioms to be zero-premise rules, so the statement in the lemma covers axioms as well. 
In Figures~\ref{fig:ruleipaxioms} and ~\ref{fig:ruleiprules} the function $\delta_R$ is defined for each rule $R$ in \GthK, Section~\ref{sec:ipfunction} explains how the figures have to be interpreted. We prove the lemma by showing that for any rule $R$, $\delta_R$ is an interpolant function. We treat rules Ax$_{\rm id}$, $\neg_l$, $\en_r$ and $R_{\fnbx}$, the other rules can be treated in a similar way.

Ax$_{\rm id}$: We consider two cases: $RL(\text{Ax}_{\rm id})$ and $RR(\text{Ax}_{\rm id})$, the other two are analogous. In the first case, $S^p = \Ga_1;\Ga_2,p \seq p,\De_1;\De_2$ and $\delta_{{\rm Ax}_{\rm id}}(S^p) =\neg p$. And indeed, $\neg p$ is an inpolant of $S^p$, as $\af \Ga_1 \seq p,\neg p,\De_1$ and $\af\Ga_2,p,\neg p \seq \De_2$. In the second case, $S^p = \Ga_1;\Ga_2,p \seq \De_1;\De_2,p$ and $\delta_{{\rm Ax}_{\rm id}}(S^p) =\top$. Since $\af \Ga_1 \seq \top,\De_1$ and $\af\Ga_2,p \seq p,\De_2$, $\top$ is indeed an interpolant for $S^p$.

$\neg_l$: Consider instance 
\begin{center}
 $\AxiomC{$\Ga\seq \phi,\De$}
 \LeftLabel{$\neg_l$}
 \UnaryInfC{$\Ga,\neg\phi \seq \De$}
 \DisplayProof$
\end{center}
and a split $S^p=(\Ga_1;\Ga_2 \seq \De_1;\De_2)$ of the conclusion. We consider the case that $\neg\phi$ belongs to $\Ga_1$, the case that it belongs to $\Ga_2$ is analogous. Thus $\Ga_1=\Pi,\neg\phi$ and the split of the premise is, by definition, $\Pi;\Ga_2 \seq \phi,\De_1;\De_2$. Suppose $\varchi$ is an interpolant of the split premise, then $\delta_{L\neg}(\varchi,S^p)=\varchi$. By assumption,  
$\af \Pi \seq \varchi, \phi,\De_1$ and $\af \Ga_2,\varchi \seq \De_2$. Clearly, the former implies $\af \Pi,\neg\phi \seq \varchi, \De_1$, which means $\af \Ga_1 \seq \varchi, \De_1$. Thus $\delta_{L\neg}(\varchi,S^p)$ is indeed an interpolant of $S^p$. 

$\en_r$: Consider instance 
\begin{center}
 $\AxiomC{$\Ga\seq \phi,\De$}
 \AxiomC{$\Ga\seq \psi,\De$}
 \LeftLabel{$\en_r$}
 \BinaryInfC{$\Ga\seq \phi \en\psi,\De$}
 \DisplayProof$ 
\end{center}
and a split $S^p=(\Ga_1;\Ga_2 \seq \De_1;\De_2)$ of the conclusion. 
First, we treat the case that the principal formula $\phi\en\psi$ belongs to $\De_1$, which means that $\De_1 = \phi \en \psi, \Sig$ for some multiset $\Sig$.
Thus the inherited splits of the left and right premise are $\Ga_1;\Ga_2 \seq \phi,\Sig;\De_2$ and $\Ga_1;\Ga_2 \seq \psi,\Sig;\De_2$. Suppose that $\varchi_1$ and $\varchi_2$ are interpolants of the respective splits: 
\begin{center}
 $\af \Ga_1 \seq \varchi_1,\phi,\Sig \ \ \ \ \af \Ga_2,\varchi_1 \seq \De_2 \ \ \ \ \ 
 \af \Ga_1 \seq \varchi_2,\psi,\Sig \ \ \ \ \af \Ga_2,\varchi_2 \seq \De_2.$
\end{center}
Using the rules of the calculus and its closure under weakening (Lemma~\ref{lem:weak}) this clearly implies that  
\begin{center}
 $\af \Ga_1 \seq \varchi_1\of\varchi_2,\phi \en\psi,\Sig \ \ \ \ 
 \af \Ga_2,\varchi_1\of\varchi_2 \seq \De_2,$ 
\end{center}
proving that $\delta_{(\en_r)}(\varchi_1,\varchi_2, S^p) = \varchi_1\of\varchi_2$ is an interpolant of $S^p$. 

Second, we treat the case that $\phi\en\psi$ belongs to $\De_2$, so that $\De_2 = \phi \en \psi, \Sig$ for some multiset $\Sig$. Thus the inherited splits of the left and right premise are $\Ga_1;\Ga_2 \seq \De_1;\phi,\Sig$ and $\Ga_1;\Ga_2 \seq \De_1;\psi,\Sig$. For any interpolants $\varchi_1$ and $\varchi_2$ of the respective splits we have 
\begin{center}
 $\af \Ga_1 \seq \varchi_1,\De_1 \ \ \ \ \af \Ga_2,\varchi_1 \seq \phi,\Sig \ \ \ \ \ 
 \af \Ga_1 \seq \varchi_2,\De_1 \ \ \ \ \af \Ga_2,\varchi_2 \seq \psi,\Sig$. 
\end{center}
Therefore, again using the rules of the calculus and its closure under weakening,  
\begin{center}
 $\af \Ga_1 \seq \varchi_1\en\varchi_2,\De_1 \ \ \ \ 
 \af \Ga_2,\varchi_1\en\varchi_2 \seq \phi\en\psi,\Sig$. 
\end{center}
This shows that $\delta_{(\en_r)}(\varchi_1,\varchi_2, S^p) = \varchi_1\en\varchi_2$ is an interpolant of $S^p$.

$R_{\fnbx}$: 
Consider instance 
\begin{center}
 $\AxiomC{$\Ga \seq \phi$}
 \LeftLabel{$R_{\fnbx}$} 
 \UnaryInfC{$\Pi,\bx\Ga \seq \bx\phi,\De$}
 \DisplayProof$ 
\end{center}
Let $S^p=(\Pi_1,\bx\Ga_1;\Pi_2,\bx\Ga_2 \seq \De_1;\De_2)$ be a split of the conclusion. There are two cases, depending on whether $\bx\phi$ belongs to $\De_1$ or $\De_2$.
The second case is straightforward: $\De_2 = \bx\phi,\Sig$ for some $\Sig$ and the split premise is $\Ga_1;\Ga_2 \seq \ ;\phi$. Thus for any interpolant $\varchi$ of the split premise, $\Ga_1 \seq \varchi$ and $\Ga_2,\varchi \seq \phi$  are derivable, and an application of $R_{\bx}$ to each of them shows that $\bx\varchi$ is an interpolant of $S^p$. 

We turn to the first case, where the split $S^p$ and the inherited split of the premise of $R_{\bx}$ have respective forms 
\begin{center}
 $\Pi_1,\bx\Ga_1;\Pi_2,\bx\Ga_2 \seq \bx\phi,\Sig ; \De_2 
 \ \ \ \  
 \Ga_1;\Ga_2 \seq \phi;\ .$
\end{center} 
For any interpolant $\varchi$ of the split premise we have  
$\af \Ga_1 \seq \varchi,\phi$ and $\af \Ga_2,\varchi \seq \ $. 
Hence $\af \Ga_1, \neg\varchi \seq \phi \ $ and $\af \Ga_2 \seq \neg\varchi $. An application of $R_{\bx}$ gives 
\begin{center}
 $\af \Pi_1,\bx\Ga_1,\bx\neg\varchi \seq \bx\phi,\Sig  
 \ \ \ \ 
 \af \Pi_2,\bx\Ga_2 \seq \bx\neg\varchi, \De_2$.
\end{center}
This implies $\af\Pi_1,\bx\Ga_1\seq \neg\bx\neg\varchi,\bx\phi,\Sig$ and $\af\Pi_2,\bx\Ga_2,\neg\bx\neg\varchi \seq \De_2$, as desired. 
\end{proof}

\begin{theorem} ({\it Sequent interpolation}) 
 \label{thm:seqip}\\
\GthK\ has sequent interpolation: 
If $\af S$, then any split of $S$ has a split-interpolant.
\end{theorem}
\begin{proof}
Suppose $\af S$ and consider an arbitrary split $S^p=(\Ga_1;\Ga_2 \seq \De_1;\De_2)$ of $S$. We use induction to the depth $d$ of the proof of $S$ in \GthK\ to show that $S^p$ has a split-interpolant, using 
Lemma~\ref{lem:ruleip} which proves that for any rule $R$ in \GthK, $\delta_R$ is an interpolant function for $R$.

In case $S$ is an instance of an axiom $R$, since $\delta_R$ is an interpolant function for $R$, $\delta(S^p)$ is an interpolant for $S^p$, and we are done. For $d>1$ we consider the last rule of the proof, $R$, and the split of the premises according to Figure~\ref{fig:ruleiprules}, to which we apply the induction hypothesis. This gives interpolant(s) $\varchi$ ($\varchi_1, \varchi_2$) of the split premises. By the definition of interpolant functions, $\delta_{R}(\varchi,S^p)$ ($\delta_R(\varchi_1, \varchi_2,S^p)$) is an interpolant of $S^p$.
\end{proof}

\begin{theorem} ({\it Interpolation for \K }) 
 \label{thm:kcip}\\
The modal logic \K\ has the Craig interpolation property. 
\end{theorem}
\begin{proof}
Suppose $\af_K \phi \imp \psi$. Then $\af_{\GthK}\phi\seq \psi$. Consider split 
$(\phi ; \ \seq \ ; \psi)$. By Theorem~\ref{thm:seqip}, calculus \GthK\ has sequent interpolation, thus there exists an interpolant $\varchi$ in the common language of $\phi$ and $\psi$ such that 
$\af_{\GthK}\phi\seq \varchi$ and $\af_{\GthK}\varchi \seq \psi$. Thus $\varchi$ is an interpolant for $\phi\imp\psi$ in \K. 
\end{proof}

\subsection{Harvest from sequent interpolation}
 \label{sec:harvest}
The constructive proof that \GthK\ has sequent interpolation, Theorem~\ref{thm:seqip}, allows us to infer further strengthenings of the theorem. On the basis of the proof it is, for example, easy to observe that any interpolant of a sequent $S$ does not contain a modal operator unless $S$ does.

A less trivial and more important consequence of the constructive proof is that it shows that \K\ has {\it Lyndon interpolation} by showing that \GthK\ has  {\it sequent Lyndon interpolation}. 
A logic has {\it Lyndon interpolation} if for any valid implication $\phi\imp\psi$ there exists an interpolant $\varchi$ such that any propositional letter that occurs positively (negatively) in $\varchi$ occurs positively (negatively) in both $\phi$ and $\psi$ (the definition of positive and negative occurrences in formulas and sequents can be found in Section~\ref{sec:preliminaries}). 
A calculus has {\it sequent Lyndon interpolation} if for any split $\Ga_1;\Ga_2 \seq \De_1;\De_2$ of any derivable sequent $S$, there is a sequent interpolant $\varchi$ such that any propositional letter that occurs positively (negatively) in $\varchi$ occurs positively (negatively) in both $\Ga_1\seq\De_1$ and $\Ga_2\seq\De_2$. 
Inspection of the proofs of Lemma~\ref{lem:ruleip} and Theorem~\ref{thm:seqip} shows that \GthK\ has this form of interpolation, which clearly implies that \K\ has Lyndon interpolation.  

\begin{corollary} ({\it Lyndon interpolation}) 
 \label{cor:seqlip}\\
\GthK\ has sequent Lyndon interpolation and \K\ has Lyndon interpolation. 
\end{corollary}

We conclude this section by listing some further advantages and disadvantages of the proof-theoretic method to prove interpolation: 

\pro{
The proof-theoretic method can be used to establish Craig interpolation for many modal logics with a sequent calculus similar to calculus \GthK\ for \K.}
 
\pro{
The proof-theoretic proof is modular: to establish that a certain extension of the calculus
for \K\ has sequent interpolation, it suffices to establish that the additional rules interpolate.}

\pro{
The proof-theoretic method presented here provides a bound on the size of the interpolant of a sequent $S$ in the size of the proofs of $S$ in \GthK.}

\con{For many logics no sequent calculus is available that has all the necessary properties to apply the Maehara method. }

\section{Third proof (syntactic)}
\label{sec:nabla}

Recall from \refchapter{chapter:propositional} how Craig interpolation for propositional logic can be proved syntactically via 
\emph{disjunctive normal form} (DNF):
to construct an interpolant for a propositional implication $\phi\to\psi$, we first rewrite $\phi$
into disjunctive normal form (i.e., as a disjunction of consistent conjunctions of literals) and then we simply ``drop'' all (positive and negative) occurrences of proposition letters 
$p\in\sig(\psi)\setminus\sig(\phi)$ from this DNF formula to obtain our Craig interpolant. Since converting a propositional formula to DNF incurs at most an 
exponential blowup, this yields an interpolant (in fact, a uniform interpolant) of 
exponential size. In this section, we will describe an interpolation method for modal logic that can be viewed as a modal analogue of this technique. It was used in \cite{tenCate2006:definitorially} and can be 
summarized in one sentence as follows: \emph{to obtain an interpolant for $\phi\to\psi$, we bring $\phi$ in a suitable normal form, and then drop all occurrences of 
proposition letters  $p\in\sig(\phi)\setminus\sig(\psi)$ from $\phi$.} 

The precise normal form that we need is known as \emph{$\nabla$-normal form}~\cite{barwise1996:vicious,janin1995:automata}.
In order to define it, we must first introduce the $\nabla$ operator. It takes a finite set of formulas as argument, and its semantics is as follows:
$M,w\models\nabla\Phi$ holds if \emph{each $\phi\in\Phi$ is true in some successor and each successor satisfies some $\phi\in\Phi$}. In other words, $\nabla\Phi$ is equivalent to
$\bigwedge_{\phi\in\Phi}\Diamond\phi\land\Box\bigvee_{\phi\in\Phi}\phi$. Note that $\nabla\emptyset$ is equivalent
to $\Box\bot$.

A modal formula is in $\nabla$-normal form if it generated by the following grammar:
\[\phi ~~::=~~ \top ~\mid~ \bot ~\mid~ \pi ~\mid~ \nabla\Phi ~\mid~ \pi\land \nabla\Phi ~\mid~ (\phi\lor\phi)\]
where $\pi$ is a consistent conjunction of literals (in other words, $\pi$ is of the form  $(\neg)p_1 \land \cdots \land (\neg) p_n$, with $n> 0$, where $p_1, \ldots, p_n$ are distinct proposition letters),
and $\Phi$ is a finite set of formulas in $\nabla$-form.

\begin{example}
    The modal formula $\Diamond p\land \Box(p\lor q)$ can be written in $\nabla$-normal form as $\nabla\{p\}\lor\nabla\{p, q\}$.
Similarly, the modal formula $\Diamond p\land \Box(q\lor r)$ can be written in $\nabla$-normal form as $\nabla\{p\land q, q\}\lor\nabla\{p\land q, q,r\}\lor\nabla\{p\land r, r\}\lor\nabla\{p\land r, q, r\}$. Finally,
the modal formula $\Box(p_1\lor p_2)\land \Box(q_1\lor q_2)$ can be written in $\nabla$-normal form as \[\bigvee_{\Phi\subseteq \{p_i\land q_j\mid i,j\in\{1,2\}\}}\nabla\Phi.\]
\eoe
\end{example}

\begin{theorem}
Every modal formula  $\phi$ is equivalent to a formula in $\nabla$-normal 
form, which can be computed from $\phi$ in exponential time.
\end{theorem}
\begin{proof}
    We will use a normalization procedure that was given in~\cite{tenCate2006:definitorially}.
    Let $\phi$ be any modal formula. We will assume that $\phi$ is in
    negation normal from. For convenience in what follows we will treat
    conjunction as an operator that can take any number of arguments.
    \[\begin{array}{llll}
    \nf(\phi) &=& \phi & \text{for $\phi$ of the form $(\neg)p, \top, \bot$}\\
    \nf(\Diamond\phi) &=& \nabla\{\nf(\phi),\top\} \\
    \nf(\Box\phi) &=& \nabla\{\nf(\phi)\}\lor \nabla\emptyset \\
    \nf(\phi\lor\psi) &=& \nf(\phi)\lor \nf(\psi) \\
    \end{array}\]
    The only difficult case is for conjunction, i.e., if
    $\phi$ is of the form $\bigwedge\Phi$. If $\Phi$ contains a formula of the form $\top, \bot$ or a 
    disjunction or conjunction, then we can apply the following simplification rules:
    \[\begin{array}{llll}
    \nf(\bigwedge(\Phi\cup\{\top\})) &=& \nf(\bigwedge\Phi) \\
    \nf(\bigwedge(\Phi\cup\{\bot\})) &=& \bot \\
    \nf(\bigwedge(\Phi\cup\{\phi\lor\psi\})) &=& \nf(\bigwedge(\Phi\cup\{\phi\}))\lor \nf(\bigwedge(\Phi\cup\{\psi\})) \\
    \nf(\bigwedge(\Phi\cup\{\bigwedge\Psi\})) &=& \nf(\bigwedge(\Phi\cup\Psi))     \end{array}\]

    This leaves us with the case where
    $\Phi$ consists of literals and formulas whose main 
    connective is a $\Diamond$ or $\Box$. 
    Let $\Phi_{lit}$ be the 
    set of literals in $\Phi$, let $\Phi_{\Diamond}$ be the set 
    of formulas $\psi$ for which $\Diamond\psi\in\Phi$ and let 
    $\Phi_\Box$ be the set of formulas $\psi$ for which $\Box\psi\in \Phi$.
    We distinguish three cases:
    \begin{enumerate}
        \item

    If $\Phi_{lit}$ contains a proposition letter and its negation,
    then clearly $\bigwedge\Phi$ is inconsistent, and we let $\nf(\bigwedge\Phi)=\bot$. 

    \item Otherwise,     if $\Phi_{\Diamond}=\emptyset$ we define
    $\nf(\bigwedge\Phi)$ as
    $\bigwedge\Phi_{lit}\land (\nabla\{\nf(\bigwedge\Phi_\Box)\}\lor \nabla\emptyset)$

    \item Otherwise, we define $\nf(\bigwedge\Phi)$
    as 
    $\bigwedge\Phi_{lit}\land \nabla(\{\nf(\psi\land\bigwedge\Phi_{\Box})\mid \psi\in \Phi_\Diamond\}\cup\{\bigwedge\Phi_{\Box}\})$
    \end{enumerate}    
    It can be verified that $\nf(\phi)$ is in $\nabla$-normal form, $\nf(\phi)$
    is equivalent to $\phi$,  
    $|\nf(\phi)| = 2^{O(|\phi|^2)}$,
    and this size bound on $\nf(\phi)$ also 
    bounds the time needed to compute $\nf(\phi)$.
\end{proof}

Intuitively, what makes the $\nabla$-normal form special is that formulas can only contain a
limited, ``interaction-free'' form of conjunction.
This is reflected in the remarkable fact (which will not be used in our proof of Craig interpolation)
that the satisfiability of a modal formula in $\nabla$-normal form can be decided in linear time by a simple recursive procedure.


As we will see in Section~\ref{sec:automata}
(cf.~Remark~\ref{rem:automata-nabla}), 
modal formulas in $\nabla$-normal form 
correspond in an intuitive sense to non-deterministic tree automata.

Next, we define a ``removal'' operator
$\tilde\exists$ that takes a modal formula
in $\nabla$-normal form and removes all occurrences
of a specified set of proposition letters. Formally,
given a formula $\phi$ in $\nabla$-normal form
and a set $\mathcal{P}$ of proposition letters, 
we define $\tilde\exists\mathcal{P}.\phi$ as follows:

\[\begin{array}{lll}
\tilde\exists\mathcal{P}.\top &=& \top \\
\tilde\exists\mathcal{P}.\bot &=& \bot \\
\tilde\exists\mathcal{P}.(\pi) &=& \pi' \\
\tilde\exists\mathcal{P}.(\pi\land\nabla\Psi) &=& \pi'\land\nabla\{\tilde\exists\mathcal{P}.\psi\mid \psi\in\Psi\} \\
\tilde\exists \mathcal{P}(\phi\lor\psi) &=&
\tilde\exists \mathcal{P}(\phi) \lor \tilde\exists \mathcal{P}(\psi)
\end{array}\]
where $\pi'$ is obtained from $\pi$ by dropping all
(positive and negative) occurrences of proposition letters in $\mathcal{P}$.

\begin{proposition}\label{prop:bisim-quant}
 For all modal formulas $\phi$ in $\nabla$-normal form over propositional signature $\sigma$, sets of proposition letters $\mathcal{P}$, and pointed Kripke models $(M,w)$,
 the following are equivalent:
 \begin{enumerate}
     \item $M,w\models\tilde\exists\mathcal{P}.\phi$,
     \item There is a pointed Kripke model $N,v\models\phi$ such that $(M,w) \bisim_{\sigma\setminus \mathcal{P}} (N,v)$.
 \end{enumerate}
\end{proposition}

\begin{proof}
    By induction on $\phi$.
    The only interesting cases in the induction argument are for formulas $\phi$ of the form $\pi$ or $\nabla\Psi$ or $\pi\land\nabla\Psi$. Of these three cases, we will discuss the latter (the former being simpler). Therefore, let 
    $\phi = \pi\land\nabla\Psi$.
    Recall that by definition,
    $\tilde\exists\mathcal{P}.\phi = \pi'\land\nabla\{\tilde\exists\mathcal{P}.\psi\mid \psi\in\Psi\}$.

    From 1 to 2:
    let $M,w\models\tilde\exists\mathcal{P}.\phi$.
    Then, for each $\psi_i\in\Psi$, there is a successor $w_i$ of $w$ such
    that $M,w_i\models\tilde\exists\mathcal{P}.\psi_i$. By induction, there is a $N_i,v_i\models\psi_i$
    such that $(M,w_i) \bisim_{\sigma\setminus \mathcal{P}} (N_i,v_i)$. 
    Similarly, for each successor $u_j$ of $w$ there is some
    $\psi_j\in\Psi$ such that $M,u_j\models\tilde\exists\mathcal{P}.\psi_j$, and, by 
    induction hypothesis, there is some $N'_j,v'_j\models\psi_j$ such
    that $(M,u_j) \bisim_{\sigma\setminus \mathcal{P}} (N'_j,v'_j)$. 
    Let $N,v$ be obtained by 
    taking the disjoint union of all the the aforementioned pointed Kripke models $N_i,v_i$ as well as $N'_j,v'_j$, adding
    a node $v$ with an edge from $v$ to each $v_i$ and to each $v'_j$, and making $\pi$
    true at $v$. It is easy to see that $N,v\models\pi\land\nabla\Psi$ and that $(N,v) \bisim_{\sigma\setminus \mathcal{P}} (M,w)$.

    From 2 to 1: assume that $(M,w) \bisim_{\sigma\setminus \mathcal{P}} (N,v)$ and 
    $N,v\models\phi$. In particular, each 
    successor of $v$ in $N$ satisfies some $\psi\in\Psi$,
    and each $\psi\in\Psi$ is
    satisfied in $N$ by a successor of $v$.
    It follows by the induction hypothesis (with the identity bisimulation) that
    each successor of $v$ in $N$ satisfies $\tilde\exists\mathcal{P}.\psi$ for some $\psi\in \Psi$,
    and for each $\psi\in\Psi$ there is a successor of $v$ in $N$ that satisfies $\tilde\exists\mathcal{P}.\psi$.
    It follows that $N,v\models\tilde\exists\mathcal{P}.\phi$.
    Since $\tilde\exists\mathcal{P}.\phi$ does 
    not contain any occurrences of the proposition letters in 
    $\mathcal{P}$, it is invariant for bisimulations over the signature without the proposition letters in $\mathcal{P}$. Therefore also 
    $M,w\models\tilde\exists\mathcal{P}.\phi$. 
\end{proof}

Finally, we get our proof of interpolation:

\begin{mainproof}
Let $\phi\to\psi$ be a valid modal implication.
We first rewrite $\phi$ to 
$\nabla$-normal form. Then, we take
$\vartheta = \tilde\exists\mathcal{P}.\phi$ with 
$\mathcal{P} =  \sig(\phi)\setminus\sig(\psi)$
We claim that $\vartheta$ is a Craig interpolant.
It is clear from the construction that
$\sig(\vartheta)\subseteq\sig(\phi)\cap\sig(\psi)$.
Therefore, it remains to show that $\phi\to\vartheta$ and $\vartheta\to\psi$ are valid.
Let $M,w\models\phi$. Then, by Proposition~\ref{prop:bisim-quant} 
(choosing  $(N,v)$ to be $(M,w)$ and using the identity bisimulation) it follows that 
$M,w\models\vartheta$.
Next, let $M,w\models\vartheta$. By Proposition~\ref{prop:bisim-quant}, there is a pointed Kripke model $N,v\models\phi$ such 
that $M,w$ and $N,v$ are bisimilar w.r.t.~all proposition letters except possibly those in $\mathcal{P}$. Therefore 
$N,v\models\psi$. It follows that $M,w\models\psi$, since $\psi$ does not
contain any occurrences of proposition letters in $\mathcal{P}$. 
\end{mainproof}

Observe that this proof, in fact, yields something stronger than Craig interpolation, namely uniform interpolation.
Furthermore, Proposition~\ref{prop:bisim-quant} shows that
the ``$\tilde\exists\mathcal{P}$'' operation can be interpreted as 
expressing a form of \emph{second-order quantification modulo bisimulations}. Indeed, this is known as \emph{bisimulation quantifiers}. Bisimulation quantifiers were originally introduced in \cite{Pitts92}  in the context of intuitionistic logic, and then used in \cite{Ghilardi2002:sheaves} and \cite{Visser1996:uniform} as a tool to prove uniform interpolation for modal logic. It was subsequently used to prove uniform interpolation for the modal $\mu$-calculus \cite{DAgostino2000:logical} and to prove 
a restricted form of uniform interpolation for the guarded fragment~\cite{DAgostino2015:guarded}.
It remains unclear to what extent the method can be adapted to other logics. For instance:
\begin{question}
    Can the above $\nabla$-normal form based method be adapted and used to prove uniform interpolation for the basic tense logic $K_t$?
\end{question}

In summary:

\pro{This proof, in fact, yields something stronger than Craig interpolation, namely uniform interpolation. }

\pro{Since 
normalization to $\nabla$-normal form can be performed in exponential time, it yields uniform interpolants of single exponential size.}

\pro{Lyndon interpolants can be obtained similarly by dropping only the positive (negative) occurrences 
of proposition letters that do not occur positively (negatively) in $\psi$). Note that the conversion 
to $\nabla$-normal form preserves polarities of proposition letters.}

\pro{The same construction also works for the modal $\mu$-calculus
(although the proof of the $\nabla$-normal form theorem is more
involved here, cf.~\cite{Venema:lectures}).} 

\con{On the other hand, since this technique yields uniform interpolation, it does not allow us to prove Craig interpolation for logics that lack uniform interpolation (eg., S4, or
the extension of modal logic with the global modality).}


\section{Fourth proof (automata-theoretic)}
\label{sec:automata}

A well-known consequence of bisimulation-invariance is the fact that
modal logic has the \emph{tree model property}: when studying properties of modal formulas, one can safely restrict attention to tree models.
\footnote{In the case of the modal logic \logic{K}, one may even restrict attention to finite trees, although this is not relevant here.}
This opens the possibility of using \emph{tree automata}. Indeed, many results in modal logic have been obtained with the help of tree automata. 
Craig interpolation, in particular, can be proved via tree automata. The high-level outline of this approach is as follows: we first define
a class of tree automata that have the same expressive power as modal logic, allowing us to translate back and forth between modal formulas and automata. Given a valid modal implication $\models\phi\to\psi$,
we then translate the formulas in question to 
tree automata, we construct a Craig interpolant
as a tree automaton, by making use of the automata-theoretic operation of \emph{projection}, and finally we translate the automaton back to a modal formula. 
As we will see (cf.~Remark~\ref{rem:automata-nabla}), 
this techniques turns out to be closely related to 
the syntactic approach from Section~\ref{sec:nabla}.

We present this approach in detail, here, for
the basic modal logic \logic{K}. As it turns out, for this case, we only need a rather simple type of tree automaton, which we will call a \emph{modal automaton}. The input of such an automaton will 
be a tree-shaped Kripke model. Formally,
let $\tau$ be any finite propositional signature.
A \emph{$\tau$-tree} is a (possibly infinite) pointed Kripke model
$(M,w)$ such that, for every world $v$ there exists 
exactly one directed path from $w$ to $v$.
Such a $\tau$-tree is said to be \emph{fat}
if every non-root world has infinitely
many bisimilar siblings, i.e.,
for each $(v,v')\in R$, there are infinitely
many $v''$ such that $(v,v'')\in R$ and $(M,v')\bisim_\tau (M,v'')$.

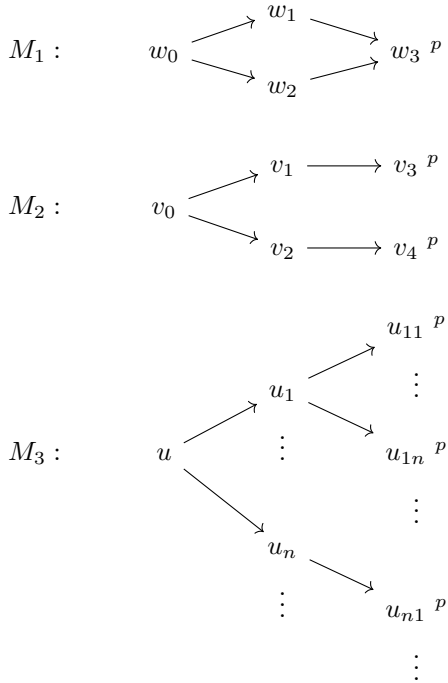
\begin{figure}[t]
\[\begin{tikzcd}[row sep = 0mm]
&                                 & w_1 \arrow[rd,""] & \\
M_1:& w_0 \arrow[ru,""] \arrow[rd,""]  &  & w_3~^p\\
&                                 & w_2 \arrow[ru,""] 
\\ ~ \\ ~ \\
& & v_1 \arrow[r,""] & v_3 ~^p\\
M_2:&v_0 \arrow[ru,""] \arrow[rd,""]  &   \\
&& v_2 \arrow[r,""] & v_4 ~^p
\\ ~ \\ ~ \\
&&& u_{11} ~^p \\
&                    &                              u_{1} \arrow[ru,""] \arrow[rd,""]  & \vdots \\
M_3: & u \arrow[ru,""]\arrow[rdd,""] & \vdots & u_{1n} ~^p   \\
&&& \vdots \\
&                    &                                u_{n}  \arrow[rd,""]\\
&                    & \vdots& u_{n1} ~^p \\
&&&\vdots
\end{tikzcd}\]
\caption{Three examples of Kripke models}
\label{fig:example-models-tree}
\end{figure}

\begin{example}
    Consider the three  Kripke models over propositional signature $\tau=\{p\}$ depicted in Figure~\ref{fig:example-models-tree}.
The pointed Kripke model $(M_1,w_0)$ is \emph{not} a $\tau$-tree, whereas
$(M_2,v_0)$ is a non-fat $\tau$-tree, and 
$(M_3,u)$ is a fat $\tau$-tree.
All three pointed Kripke models are $\tau$-bisimilar.
\end{example}

The following proposition explains why we may 
restrict attention to tree-shaped Kripke models.

\begin{proposition}
    Let $(M,w)$ be a pointed Kripke model 
    over propositional signature $\tau$. Then there is a
    fat $\tau$-tree, which we will denote by $(unr(M),\langle w\rangle)$ such that $(M,w)\bisim_\tau (unr(M), \langle w\rangle)$.
\end{proposition}
\begin{proof}
    The proof uses a variant of the standard tree unraveling construction, cf.~\cite{BlackburnDeRijkeVenema}:
    let $M=(W,R,V)$. By a \emph{path from $w$} we will 
    mean a finite sequence 
    $\langle w_0, k_1,w_1,k_2,w_2\ldots, k_{n}, w_n\rangle$ ($n\geq 0$) where $w_0=w$,
    where $\langle w_i,w_{i+1}\rangle\in R$ for all $i<n$, and where each $k_i$ is a natural number.
    By a \emph{successor} of a path  $\langle w_0, \ldots, w_n\rangle$ we will mean a path that extends it with one more number and world, or, in other words,     a path of the form 
    $\langle w_0, \ldots, w_n, k_{n+1}, w_{n+1}\rangle$. 
    Let $unr(M)=(W',R',V')$ where
    $W'$ is the set of all paths from $w$,
    $R'$ is the successor relation on such paths,
    and $V'(p)$ is the set of all paths whose
    final world element belongs to $V(p)$. It can be
    verified that $(unr(M), \langle w\rangle)$ is a fat $\tau$-tree and that $(M,w)\bisim_\tau (unr(M), \langle w\rangle)$. Indeed, the relation $Z$
    that 
    consists of all pairs $(x,v)$ with
    $x$ a path and $v$ its final world element
    is a $\tau$-bisimulation.
\end{proof}

For a propositional signature $\tau$, we will denote by $2^\tau$ the set of all truth assignments for 
the proposition letters in $\tau$. To simplify notation,
we will equate each truth assignments for $\tau$ with the subset of $\tau$ consisting of the 
proposition letters that are assigned to true.
For a world $w$ in a given Kripke model, $M=(W,R,V)$, we will denote
by $lab_\tau(w)=\{p\in\tau\mid w\in V(p)\}$ the  truth assignment that
makes true precisely those $p\in\tau$ for which
it holds that $w\in V(p)$. This will allow us 
to view a $\tau$-tree as a node-labeled tree
(where $2^\tau$ is the set of node labels).

\begin{definition}[Modal automaton]
Let $\sigma$ be any finite propositional signature.
A modal automaton\footnote{The name ``modal automaton'' we have chosen here, may be considered slightly misleading: as we will see in a moment (cf.~Example~\ref{ex:modal-automaton}), these automata are strictly more expressive than modal formulas. We will be adding a further restriction to obtain an exact equivalence.} for $\sigma$ 
is a tuple $A=(\Sigma,Q,\delta,q_0,F)$
consisting of 
\begin{itemize}
\item An \emph{alphabet} $\Sigma = 2^\sigma$,
\item A finite set of \emph{states} $Q$,
\item A \emph{transition relation} $\delta\subseteq Q\times \Sigma\times \wp(Q)$,
\item An \emph{initial state} $q_0\in Q$, and
\item A set of \emph{accepting states} $F\subseteq Q$.
\end{itemize}
\end{definition}

\begin{remark}
    As customary in automata theory, we will use the letter $Q$ to denote the set of states of the automaton and we will use the letter $q$ to denote an element of $Q$, that is, a state of the automaton. Therefore, in the remainder of this section, $q$ will always denote a state (unlike elsewhere in this chapter, where we have been using $q$ occasionally as a proposition letter).
\end{remark}

In the same way that a modal formula evaluates to 
true or false in a given pointed Kripke model $(M,w)$,  a modal automaton can either \emph{accept} or 
\emph{reject} a given $\tau$-tree.

\begin{definition}[Acceptance]
An \emph{accepting run} of a modal automaton $A=(\Sigma,Q,\delta,q_0,F)$ for $\tau$ on $\tau$-tree $(M,w)$ with $M=(W,R,V)$ is a function $\rho:W\to Q$ such that the following conditions are met:
\begin{enumerate}
    \item $\rho(w)= q_0$,
    \item whenever $\rho(v) = q$, either $q\in F$
or $(q,lab_\tau(v),\{\rho(v')\mid (v,v')\in R\})\in\delta$.
\end{enumerate}
We say that $(M,w)$ is accepted by $A$ if there is an accepting run of $A$ on $(M,w)$.
\end{definition}

Note that a run of a modal automaton may accepting even if an accepting state does not appear on every branch.
Indeed, the final states can be dispensed with entirely, by introduce a special state $q^*$ together with all transitions of the form $(q_f,\ell,\{q^*\}$ for $q_f\in F$ and $a\in\Sigma$, as well as all transitions of the form $(q^*,a,\{q^*\}$ and $(q^*,a,\emptyset\}$. However, having an explicit set of accepting states in the specification of automata will help us define ``acyclicity'' below. 
\begin{example}\label{ex:modal-automaton}
Consider the propositional signature $\tau=\{p\}$.
Consider the automaton $A$ with two states, $q_0, q_1$, 
where $q_0$ is the initial state and $q_1$ is an accepting state, that has transitions
$(q_0,\emptyset,\{q_0\})$, $(q_0,\{p\}, \emptyset)$ and $(q_0,\{p\},\{q_1\})$.

Let $(M,w_1)$ be the $\tau$-tree depicted as follows:
\[\begin{tikzcd}[row sep = -1mm]
&&& p \\
                    &                                  & w_2 \arrow[r,""] & w_4 \\
w_0 \arrow[r,""] & w_1 \arrow[ru,""] \arrow[rd,""]  &   \\
                    &                                  & w_3 \arrow[r,""] & w_5 \\
                    && p
\end{tikzcd}\]
Then $A$ accepts $(M,w_0)$. Indeed, the function $\rho$
given by 
\[\rho(v)=\begin{cases}
            q_0 & \text{ for $v\in\{w_0, w_1, w_2,w_3,w_4\}$} \\
            q_1 & \text{ for $v\in\{w_5\}$ }
          \end{cases}
\]
is an accepting run.
More generally, $A$ accepts precisely the $\tau$-trees in which
\emph{every path ending in a leaf contains a node that satisfies $p$}.
There is no modal formula that defines the same property, although it can be expressed in extensions of modal logic such as \logic{CTL^*} (by the formula $A(Fp\lor GF\top)$) and the modal $\mu$-calculus (by the formula $\nu x(p\lor \Box x)$).
\eoe
\end{example}

As the above example shows, there is not an exact correspondence between modal automata and modal formulas, since modal automata are too expressive.
To solve this, we will restrict attention to acyclic automata.
We say that a modal automaton is \emph{acyclic} 
if there is a function $rank:Q\to \mathbb{N}$
such that whenever $(q,a,S)\in\delta$ and $q'\in S$,
then $rank(q')<rank(q)$. If a 
modal automaton is acyclic, intuitively, this means that it does not allow for an infinite sequence of transitions. Note that the automaton in the above example is \emph{not} acyclic.

\begin{theorem}
\label{thm:automata-formulas}
Fix any finite propositional signature $\sigma$.
\begin{enumerate}
    \item
    For every modal formula $\chi$, there
    is an acyclic modal automaton $A_\chi$ such that
    $A_\chi$ accepts precisely the $\sigma$-trees that satisfy
     $\chi$. 
\item
    Conversely, for every acyclic modal automaton $A$
    there is a modal formula $\chi_A$ such that
    a fat $\sigma$-tree satisfies $\chi_A$ iff it is accepted by $A$.
\end{enumerate}
\end{theorem}

\begin{proof}
    1. We may assume without loss of generality that $\chi$ is in negation normal form.
    Let $X$ be the set of all subformulas of $\chi$. 
    For $\Psi,\Psi'\subseteq X$,
    we say that $\Psi'$ is a 
    \emph{decisive refinement} of $\Psi$ if 
    \begin{enumerate}
        \item $\Psi\subseteq\Psi'$,
        \item whenever $\chi_1\land\chi_2\in \Psi'$, then $\chi_1\in \Psi'$ and $\chi_2\in\Psi'$, 
        \item whenever $\chi_1\lor\chi_2\in \Psi'$, then $\chi_1\in\Psi'$ or $\chi_2\in\Psi'$,  and
        \item $\bot\not\in\Psi'$,     \end{enumerate}    
    We define the \emph{rank} of $\Psi\subseteq X$ to be $0$ if 
    $\Psi=\emptyset$, and $1+n$
    otherwise, where $n$ is the maximal modal depth of a formula in $\Psi$.
    
    We now construct our automaton $A_\chi$ as follows:
    \begin{itemize}
        \item $\Sigma = 2^\sigma$
        \item $Q = \{q_\Psi\mid \Psi\subseteq X \}$.
        \item $(q_\Psi,a,S)\in \delta$ iff there is a  decisive refinement $\Psi'$ of $\Psi$ such that
        \begin{enumerate}
        \item[(i)] all the literals in $\Psi'$ are satisfied under
        the propositional valuation $a$;
        \item[(ii)] for each formula in $\Psi'$ of the form $\Diamond\psi$, there is some $q_{\Psi''}\in S$ with $\psi\in\Psi''$; 
        \item[(iii)] for each formula in $\Psi'$ of the form $\Box\psi$ and for every $q_{\Psi''}\in S$,  $\psi\in\Psi''$; and
        \item[(iv)] for each $q_{\Psi''}\in S$, the  rank of $\Psi''$ is strictly less than the rank of $\Psi$.
        \end{enumerate}
        \item The initial state $q_0$ is $q_{\{\chi\}}$,
        \item $F=\{q_{\emptyset}\}$.
    \end{itemize}

    Note that this automaton is, by construction, acyclic. Furthermore, it can be shown that
    $A_\chi$ accepts precisely those 
    $\sigma$-trees that satisfy
    $\chi$. More precisely, let us say
    that $A_\chi$ accepts a $\sigma$-tree $(M,w)$ \emph{from a state $q_\Psi$} if it accepts $(M,w)$
    when using $q_\Psi$ (instead of 
    $q_{\{\chi\}}$ as the initial state.
    Then, it can be shown that
    \emph{$A_\chi$ accepts a $\sigma$-tree $(M,w)$ from a state $q_\Psi$ if and only if 
    $M,w\models\psi$ for all $\psi\in\Psi$.}
    The proof, which is by induction on the rank of $\Psi$, is left as an exercise to the reader.

    2. Let $A=(\Sigma,Q,\delta,q_0,F)$ be an acyclic modal automaton, with $\Sigma=2^\sigma$. Let $rank:Q\to\mathbb{N}$ be a witness  for the fact that $A$ is acyclic.
    We will construct a modal formula $\chi_q$ for each $q\in Q$, by induction on the rank of $q$. If 
    $q\in F$, we can pick
    $\chi_q$ to be $\top$. Otherwise, we can define $\chi_q$ as follows (by induction on rank):
    \[\chi_q := \bigvee_{(q,a,S)\in\delta}\Big(\bigwedge_{p\in a}p\land\bigwedge_{p\in\sigma\setminus a} \neg p\land \nabla\{\chi_{q'}\mid q'\in S\}\Big)\]
    Here, we conveniently make use of the $\nabla$ operator that we defined in Section~\ref{sec:nabla} (which can also be written out in terms of $\Diamond$ and $\Box$).
    
    \medskip\par\noindent\emph{Claim: }
    A fat $\sigma$-tree $(M,w)$ satisfies 
    $\chi_q$  if and only if it is
    accepted by $A$ from state $q$. 

    \medskip\par\noindent
    The claim can be proved by induction on the 
    rank of the state $q$. For the induction
    step, the ``if'' direction is entirely straightforward, while the ``only if'' direction relies on the following fact about fat $\sigma$-trees: whenever $\nabla\{\phi_1, \ldots, \phi_n\}$ is satisfied by at a world $w$
    in a fat $\sigma$-tree, then there are 
    \emph{distinct} successors $v_1, \ldots, v_n$
    of $w$ such that each $v_i$ satisfies $\phi_i$.

    It follows from the above claim that $\chi_A=\chi_{q_0}$ is
    satisfied by those fat $\sigma$-trees that are accepted by $A$. 
\end{proof}

One of the fundamental insights underlying much of automata theory, is that the languages accepted by finite-state automata are closed under natural operations such as \emph{intersection},
\emph{complementation}, and \emph{projection}. Of these operations, \emph{projection} turns out to be closely related to (uniform) Craig interpolation. For modal automata, the projection operation is most naturally defined as follows:

\begin{definition}[Projection]
Let $A=(2^{\sigma},Q,\delta,q_0,F)$ be a modal automaton over signature $\sigma$,
and let $\tau\subseteq \sigma$. The \emph{projection of $A$ to $\tau$}, denoted
$Proj_\tau(A)$, is the modal automaton
$A=(2^{\tau},Q,\delta',q_0,F)$ where
$\delta = \{(q,a\cap\tau,S)\mid (q,a,S)\in \delta\}$.
\end{definition}

For a Kripke model $M=(W,R,V)$ over a propositional signature
$\sigma$ and for $\tau\subseteq\sigma$, by
the \emph{$\tau$-reduct} of $M$ we will mean
the Kripke model over signature $\tau$
obtained by restricting $V$ to the proposition 
letters in $\tau$.

\begin{proposition}\label{prop:projection}
    Let $A$ be a modal automaton over signature $\sigma$ and let $\tau\subseteq\sigma$.
    Then $Proj_\tau(A)$ accepts precisely 
    those $\tau$-trees that are the $\tau$-reduct of a  $\sigma$-tree accepted by $A$.
\end{proposition}

\begin{proof}
    Suppose $A$ accepts $(M,w)$, and let $\rho$
    be a witnesssing accepting run of $A$ for $(M,w)$. It follows
    immediately from the construction of $Proj_\tau(A)$ that the same function $\rho$
    is also an accepting run of $Proj_\tau(A)$ for the $\tau$-reduct of $(M,w)$.

    Conversely, suppose that $Proj_\tau(A)$ accepts
    a $\sigma$-tree $(M,w)$, where $M=(W,R,V)$.
    Let $\rho$ be a witnessing accepting run.
    By definition, this means that, for each
    $v\in W$ with $\rho(v)\not\in F$, the triple of the form $(\rho(v),lab_\tau(w),\{\rho(u)\mid (v,u)\in R\})$ belongs to the transition relation of $Proj_\tau(A)$. The transition relation of $A$ must therefore contain a
    triple $(\rho(v),a_v,\{\rho(u)\mid (v,u)\in E\})$ such that $a_v\cap\tau = lab_\tau(v)$.
    Now, choose any such triple and let $M'=(W,R,V')$ where $V':\sigma\to\wp(W)$ is given by $v\in V(p)$ iff $p\in a_v$. Note that this is only well defined for $\rho(v)\not\in F$.
    For those $v\in W$ with $q(v)\in F$, we can simply
    let $V'$ agree with $V$ on all proposition letters. By construction,
    $M$ is the $\tau$-reduct of $M'$, and
    $\rho$ is an accepting run of $A$ for
    $(M',w)$.    
\end{proof}

Proposition~\ref{prop:projection} can be viewed as an 
automata-theoretic analogue of Proposition~\ref{prop:bisim-quant}.
Note that the projection operation preserves the acyclicity of a modal automaton.

\begin{mainproof}
    Let $\models\phi\to\psi$ be a valid modal 
    implication. We construct $\vartheta$ by
    taking the modal automaton of $\phi$,
    projecting it to $\tau = \sig(\phi)\cap\sig(\psi)$, 
    and translating it back to a modal formula.
    That is, $\vartheta = \chi_{Proj_{\tau}(A_\phi)}$.
    We claim that $\vartheta$ is a Craig interpolant. It is clear from the construction that
$\sig(\vartheta)\subseteq\sig(\phi)\cap\sig(\psi)$.
Therefore, it remains to show that $\phi\to\vartheta$ and $\vartheta\to\psi$ are valid.

Let $M,w$ be any pointed Kripke model over $\sig(\phi)\cup\sig(\psi)$ such that
$M,w\models\phi$. We may assume
without loss of generality that 
$M,w$ is a fat $\sig(\phi)\cup\sig(\psi)$-tree. By Proposition~\ref{prop:projection},
the $\tau$-reduct of $M,w$ satisfies $\vartheta$
and hence $M,w\models\vartheta$.

Next, let $M,w$ be any pointed Kripke model over $\sig(\phi)\cup\sig(\psi)$ such that $M,w\models\vartheta$. Again, we may assume without loss of generality that $M,w$ is a fat $\sig(\phi)\cup\sig(\psi)$-tree. Since
$\sig(\vartheta)\subseteq\tau$, the $\tau$-reduct
of $M,w$ already satisfies $\vartheta$. 
It follows that the $\tau$-reduct
of $M,w$ is accepted by $Proj_{\tau(A_\phi)}$.
By Proposition~\ref{prop:projection}
and Theorem~\ref{thm:automata-formulas},
the $\tau$-reduct of $M,w$ is equal to the $\tau$-reduct
of some $\sig(\phi)\cup\sig(\psi)$-tree $N,v\models\phi$. In other words, $M,w$ and
$N,v$ are $\tau$-isomorphic. Let $f:M\cong N$
be a $\tau$-isomorphism. Next, we extend
$f$ to an $\sig(\psi)$-isomorphism by changing
$N$. More precisely, 
let $M=(W^M,R^M,V^M)$ and $N=(W^N,R^N,V^N)$ and let 
let $N'$ be a copy of $N$
in which the valuation of each proposition letter
$p\in \sig(\psi)\setminus\tau$ is replaced by
$\{f(u)\mid u\in V^M(p)\}$. By construction,
$f$ is a $\sig(\psi)$-isomorphism between $(M,w)$ and $(N',v)$, and $N'$ differs from $N$ only
on the interpretation of proposition letters 
that do not belong to $\tau$. 
Since $N,v\models\phi$
and $N'$ differs from $N$ only
with respect to proposition letters that do not appear in $\phi$, we have
that $N',v\models\phi$ and hence
$N',v\models\psi$.
This allows us to conclude,
finally, that $M,w\models\psi$.  
\end{mainproof}

\begin{remark}
    A close analysis of the proof of Theorem~\ref{thm:automata-formulas} reveals that $A_\chi$ can be constructed from $\chi$ in exponential time. Moreover a DAG representation of $\chi_A$ can be computed from $A$ in polynomial time. Since the projection operation can be performed in polynomial time (in fact, in linear time) this gives us a method for constructing interpolants (indeed, uniform interpolants in $\nabla$-normal form) of singly exponential size if we allow DAG representations.
\end{remark}

\begin{remark}
    Many different types of tree automata are in use, in the context of modal logic and elsewhere.
    These include \emph{deterministic}, \emph{non-deterministic} and \emph{alternating} automata,
    that take as input either \emph{ranked trees}
    (where each non-leaf node has a fixed number of children) or \emph{unranked-trees}. Moreover,
    the trees in question can be either \emph{finite}
    or \emph{infinite}, and in the case of infinite
    trees, a variety of different \emph{acceptance conditions} can be used to determine if a given 
    run is accepting. 
    The modal automata we used here are non-deterministic automata that run on (finite or infinite) unranked trees and that have a very simple type of acceptance condition. For more
    expressive extensions of the modal language, 
    such as the modal $\mu$-calculus, one typically
    needs more complex acceptance conditions, which
    speak about the states of the automaton that are allowed to appear infinitely often on the different branches of an infinite run.
    See~\cite{TATA} for a 
    comprehensive introduction to different types of tree automata. The use of automata-theoretic techniques in the study of modal logics goes back to~\cite{Vardi1986:automata}.
\end{remark}

\begin{remark}\label{rem:automata-nabla}
The automata-theoretic proof of Craig interpolation
we presented here is, at a technical level, very close
to the proof of 
interpolation via $\nabla$-normal form (Section~\ref{sec:nabla}). Indeed, the 
automata-theoretic proof produces a (uniform) interpolant that is in $\nabla$-normal form. 
The reader might wonder what the advantage is 
of going through automata when there is a
direct syntactic construction that already
yields a uniform interpolant of single exponential size.
The reason we have included the automata theoretic
proof is that it helps us understand better what
is happening at a conceptual level. Indeed, in 
some sense, a modal formula in $\nabla$-normal form
\emph{is} just an acyclic modal automaton
(and, moreover, this correspondence comes with
an efficient translation if we allow a DAG-representation for the formula).
\end{remark}

In light of Remark~\ref{rem:automata-nabla}, we omit a list of pros and cons, as they correspond to the ones that were already given in Section~\ref{sec:nabla}.

\section{Fifth proof (via quasi-models and type elimination sequences)}

The interpolation method we will now describe  is one that was successfully applied to the guarded-fragment and the guarded-negation fragment in~\cite{Benedikt2015:effective} to compute interpolants of worst-case optimal size.
Furthermore, in \cite{Jung2022:more} a refinement of this approach was used to obtain decision procedures 
for the existence of interpolants in some extensions of modal logic that lack Craig interpolation.
In both cases the  constructions used are quite involved. Here, we illustrate the method in the 
simpler setting of the basic modal language.

We first review quasi-models (a.k.a.~type models) and type elimination sequences as a  technique for deciding validity (or, equivalently, satisfiability). This technique can be traced back to~\cite{Pratt1979} where it was used to show that the satisfiability problem for Propositional Dynamic Logic (PDL) is in \textsf{ExpTime}.
After this review, we will show how to extract Craig interpolants (and, in fact, Lyndon interpolants) from type elimination sequences.

A quasi-model is, essentially, a consistent
collection of ``types'', where a type is a subset $X\subseteq Y$. Here, $Y$
is some finite set of relevant formulas. For instance, $Y$ may be the set of all subformulas of 
a given formula whose validity we are trying to determine.

Our presentation is specifically tailored to the case where the input formula 
(for which we want to test validity) is an implication
between a pair of modal formulas. This will facilitate the 
interpolant-construction later. 
In addition, some of our definitions (e.g., the definition of $\subf$ below) are
tailored in anticipation of the specific use case of constructing Lyndon interpolants.

Recall the \emph{negation normal form (NNF)} for modal formulas
(cf.~Section~\ref{sec:basic}).
 
\begin{definition}[SUBF, and LITERALS] \ 
\begin{itemize}
  \item If $\chi$ is a formula in NNF,  we denote by $\subf(\chi)$ the following set:
   \[\begin{array}{lll}
   \subf(\alpha) &=& \{\alpha\} \text{ for all $\alpha$ of the form $p$, $\neg p$, $\top$ or $\bot$} \\
   \subf(\chi_1\land\chi_2) &=& \subf(\chi_1)\cup \subf(\chi_2)\cup \{\chi_1\land\chi_2\} \\
   \subf(\chi_1\lor\chi_2) &=& \subf(\chi_1)\cup \subf(\chi_2)\cup \{\chi_1\lor\chi_2\} \\
   \subf(\Diamond\chi) &=& \subf(\chi)\cup\{\Diamond\chi\} \\   
   \subf(\Box\chi) &=& \subf(\chi)\cup\{\Box\chi\} 
   \end{array}\]
   Note that $\subf(\neg p)$ does \emph{not} include p.
   \item If $\chi$ is a formula in NNF, we denote by $\lits(\chi)$ the set of all formulas in $\subf(\chi)$  of the form $p$ or $\neg p$. For example, 
   $\lits(\Box p\land\Diamond\neg q) = \{p, \neg q\}$.
   \end{itemize}
\end{definition}

Fix a modal implication $\phi\to\psi$.
The next few definitions are all relative to the given choice of the modal formulas $\phi$ and $\psi$.

\begin{definition}[Types; Overlap-Consistency] \ 
\begin{enumerate}
\item By a \emph{locally-consistent} subset of $\subf(\chi)$ we mean a subset $X\subseteq \subf(\chi)$ such that (i) whenever $\chi_1\land\chi_2\in X$, then $\chi_1, \chi_2\in X$; (ii)
  whenever $\chi_1\lor\chi_2\in X$, then at least one of $\chi_1,\chi_2$ belongs to $X$; (iii)
    $\bot\not\in X$; and (iv) it is not the case that $p, \neg p \in X$ for some proposition letter $p$.
\item An \emph{\textsc{l}-type} is a locally-consistent subset of $\subf(\nnf(\phi))$, and an \emph{\textsc{r}-type} is a  locally-consistent subset of $\subf(\nnf(\neg\psi))$. A combined \emph{type} (or just $\emph{type}$ for short) is a pair $\tau=(\tau_{\textsc{l}},\tau_{\textsc{r}})$ where $\tau_{\textsc{l}}$ is an \textsc{l}-type and $\tau_{\textsc{r}}$ is an \textsc{r}-type.
\item A  type $\tau=(\tau_{\textsc{l}},\tau_{\textsc{r}})$ is \emph{overlap-consistent} if there does not exist a proposition letter $p$ such that $p\in\tau_{\textsc{l}}$ and $\neg p\in \tau_{\textsc{r}}$ or vice versa.
\end{enumerate}
\end{definition}

It is worth pointing out that the above definition of \emph{local consistency} is closely related to the definition of \emph{decisive refinements} used in the proof of Theorem~\ref{thm:automata-formulas}. Moreover, both are variants of the concept of \emph{Hintikka sets} that is typically used in completeness proof for semantic tableau systems.

Note that the ``\textsc{r}'' in \emph{\textsc{r}-types} is not related to the accessibility relation in a Kripke model, which is coincidentally also usually denoted by $R$.

\begin{definition}[Type of a world]
Given a Kripke model $M=(W^M,R^M,V^M)$,
and $v\in W^M$, the type $\tau(v)=(\tau_{\textsc{l}}(v), \tau_{\textsc{r}}(v))$ is defined as follows: $\tau_{\textsc{l}}(v)$ is the set of formulas from $\subf(\nnf(\phi))$ true at $v$, 
and $\tau_{\textsc{r}}(v)$ is the set of formulas from $\subf(\nnf(\neg\psi))$ true at $v$.
By construction, this is an overlap-consistent type.
\end{definition}

\begin{definition}[Viable successor]
Let $\tau, \tau'$ be D-types, for $D\in\{L,R\}$.
We say that $\tau'$ is a viable successor for $\tau$ (notation: $\tau\Rightarrow \tau'$) if for every formula of the form $\Box\chi$ belonging to $\tau$, $\chi$ belongs to $\tau'$. This definition extends naturally to combined types:
we say that $\tau\Rightarrow\tau'$ holds if $\tau_{\textsc{l}}\Rightarrow \tau'_{\textsc{l}}$ and $\tau_{\textsc{r}}\Rightarrow \tau'_{\textsc{r}}$.
\end{definition}

\begin{definition}[Quasi-model] \label{def:quasi}
A \emph{quasi-model} is a set $\mathcal{X}$ of (combined) types, such that
\begin{enumerate}
    \item Every $\tau\in \mathcal{X}$ is overlap-consistent.
    \item For all $\tau\in \mathcal{X}$ and $\Diamond\chi\in \tau_D$ (with $D\in\{L,R\}$), there is $\tau'\in \mathcal{X}$ such that $\tau\Rightarrow\tau'$ and 
    $\chi\in\tau'_D$. 
\end{enumerate}
\end{definition}

Note that the empty set of types is, in particular, a quasi-model.

This completes the definition of quasi-models. 
The next theorem states that quasi-models can be used to decide validity of $\phi\to\psi$.

\begin{theorem} {\bf (Soundness and completeness of quasi-models)} \label{thm:quasi1}
The modal formula $\phi\land\neg\psi$ is satisfiable if and only if 
there is a quasi-model $\mathcal{X}$ with a type $\tau=(\tau_{\textsc{l}},\tau_{\textsc{r}})\in \mathcal{X}$ such that
$\nnf(\phi)\in\tau_{\textsc{l}}$ and $\nnf(\neg\psi)\in \tau_{\textsc{r}}$.
\end{theorem}

\begin{proof} (sketch)
In one direction, let $M=(W^M,R^M,V^M)$ be a Kripke model and suppose $M,w\models\phi\land\neg\psi$
for some $w\in W^M$.
Let $\mathcal{X} = \{\tau(v)\mid \text{$v\in W^M$}\}$. Then it is easy to show that 
$\mathcal{X}$ is a quasi-model.
Note that whenever $v R^M u$, then $\tau(u)$ is indeed a viable successor
for $\tau(v)$.

Conversely, suppose $\mathcal{X}$ is a quasi-model with a type $\tau=(\tau_{\textsc{l}},\tau_{\textsc{r}})\in \mathcal{X}$ such that $\nnf(\phi)\in\tau_{\textsc{l}}$ and $\nnf(\neg\psi)\in \tau_{\textsc{r}}$. We construct a model as follows:
for each $\tau'\in \mathcal{X}$ we create a world $w_{\tau'}$. The accessibility relation connects a pair of worlds $(w_{\tau'}, w_{\tau''})$ whenever $\tau'\Rightarrow \tau''$. Finally, a proposition letter $p$ is set to true at $w_{\tau'}$ whenever $p\in\tau'_{\textsc{l}}\cup\tau'_{\textsc{r}}$. It can be shown by induction that the resulting model $M_{\mathcal{X}}$ satisfies the following truth lemma: 
\[M_{\mathcal{X}},w_{\tau'}\models\chi \text{ for all } \chi\in\tau'_{\textsc{l}}\cup \tau'_{\textsc{r}}\]
In particular, it follows that $M_{\mathcal{X}},w_\tau\models \phi\land\neg\psi$.
\end{proof}

Clearly, a quasi-model is a finite object. More precisely, since a type is a 
polynomial-sized object, a quasi-model is an object of singly exponential size
(as a function of the size of $\phi$ and $\psi$).
Therefore, Theorem~\ref{thm:quasi1} provides us with a decision procedure for testing satisfiability of $\phi\land\neg\psi$. The immediate
upper bound we get from this theorem is \textsf{NExpTime} (by non-deterministically 
guessing the quasi-model). In fact, we can do a little better.

\begin{definition}[Type elimination sequence]
A \emph{type elimination sequence} is a sequence $\mathcal{X}_0, \ldots, \mathcal{X}_n$
where 
\begin{enumerate}
    \item $\mathcal{X}_0$ is the set of all types (for the given formulas $\phi, \psi$), 
    \item each $\mathcal{X}_{i+1}$ is obtained from $\mathcal{X}_i$ by removing a type $\tau\in \mathcal{X}_i$ that fails to satisfy condition 1 or 2 from the definition of quasi-models.
    \item $\mathcal{X}_n$ is a quasi-model.
    \end{enumerate}
\end{definition}

It is easy to see that a type elimination sequence always exists
(note that $\mathcal{X}_0$ is finite and that $\mathcal{X}_n$ may be empty). Even if, for a given choice of $\phi$ and $\psi$, there does not exist \emph{any} type, there is still a type elimination
sequence, namely where $\mathcal{X}_0=\mathcal{X}_n=\emptyset$.

\begin{theorem} \label{thm:quasi2} \ 
\begin{enumerate}
\item All type elimination sequences (for the given formulas $\phi, \psi$) end in the 
same quasi-model $\mathcal{X}_n$, which can equivalently be characterized as the maximal quasi-model, and as the union of all quasi-models.
    \item The modal formula $\phi\land\neg\psi$ is satisfiable if and only if $\mathcal{X}_n$ contains a type $\tau=(\tau_{\textsc{l}},\tau_{\textsc{r}})$
with $\nnf(\phi)\in \tau_{\textsc{l}}$ and $\nnf(\neg\psi)\in\tau_{\textsc{r}}$.
\end{enumerate}
\end{theorem}

\begin{proof}
For Item 1, it suffices to observe that, if $\mathcal{X}_0, \ldots, \mathcal{X}_n$ is a
type elimination sequence and $\mathcal{X}$ is any quasi-model, then, $\mathcal{X}\subseteq \mathcal{X}_i$ for all $i\leq n$.
Indeed, this can be shown by a straightforward induction on $i$. 

For Item 2, if $\phi\land\neg\psi$ is satisfiable, then let $\mathcal{X}$ be the quasi-model given by Theorem~\ref{thm:quasi1}. By Item 1, we have that $\mathcal{X}\subseteq \mathcal{X}_n$, and hence, 
$\mathcal{X}_n$ contains a type $\tau=(\tau_{\textsc{l}},\tau_{\textsc{r}})$
with $\nnf(\phi)\in \tau_{\textsc{l}}$ and $\nnf(\neg\psi)\in\tau_{\textsc{r}}$.
If, on the other hand, $\phi\land\neg\psi$ is \emph{not} satisfiable, then it follows from Theorem~\ref{thm:quasi1} that there is no quasi-model containing a type $\tau$ with $\nnf(\phi)\in \tau_{\textsc{l}}$ and $\nnf(\neg\psi)\in\tau_{\textsc{r}}$. In particular, $\mathcal{X}_n$ does not contain such a type.
\end{proof}

This puts the complexity in \textsf{ExpTime}: it suffices to construct an arbitrary type elimination sequence and inspect the final type set of the sequence. Note that
the length of the sequence is single exponential because one type gets eliminated
each step.
Note that this complexity upper bound is still not optimal, because the satisfiability problem for the modal logic \logic{K} is in \textsf{PSpace}. However, the quasi-model method is quite generic and can be adapted to various extensions of modal logic that are \textsf{ExpTime}-complete, such as the extension with the global modality, or Propositional Dynamic Logic (PDL).

Next, we explain how to construct interpolants from type elimination sequences. In fact we will show how to construct Lyndon interpolants.
Recall that a \emph{Lyndon interpolant} for a valid implication $\phi\to\psi$ is a formula $\vartheta$
such that $\phi\to\vartheta$ and $\vartheta\to\psi$ are valid, and such that every proposition
letter occurring positively (negatively) in $\vartheta$ occurs positively (negatively) in both
$\phi$ and $\psi$. In particular, a Lyndon interpolant is a Craig interpolant (but not vice versa).

Fix modal formulas $\phi, \psi$ such that $\phi\to\psi$ is valid.
By Theorem~\ref{thm:quasi2} there is a type elimination sequence $\mathcal{X}_0, \ldots, \mathcal{X}_n$
such that $\mathcal{X}_n$ does not contain any type $\tau$ with $\nnf(\phi)\in\tau_{\textsc{l}}$ and 
$\nnf(\neg\psi)\in \tau_{\textsc{r}}$. For the remainder of this section, we can fix such a sequence.

The core result is:

\begin{theorem} \label{thm:main}
If a type $\tau=(\tau_{\textsc{l}},\tau_{\textsc{r}})$ gets eliminated in the sequence, then
there is a modal formula $\vartheta_\tau$, such that
\begin{itemize}
    \item 
$\models (\bigwedge\tau_{\textsc{l}})\to\vartheta_\tau$ and
\item $\models\vartheta_\tau\to\neg(\bigwedge\tau_{\textsc{r}})$ and
\item Every proposition letter occurring positively (negatively) in $\vartheta_\tau$
occurs positively (negatively) in both $\phi$ and $\psi$.
\end{itemize}
\end{theorem}

\begin{proof}
The proof is by induction on the stage at which the type gets eliminated.
If a type $\tau$ is eliminated, this is for one of two reasons:

1.  $\tau$ is not overlap-consistent. If $p\in \tau_{\textsc{l}}$ and $\neg p\in\tau_{\textsc{r}}$, we pick $\vartheta_\tau = p$. 
Note that $p\in \lits(\phi)$ and $\neg p\in \lits(\nnf(\neg\psi))$. 
Since polarity of proposition letter occurrences is preserved when bringing formulas into negation normal form, and polarity of proposition letter occurrences is inverted when negating a formula, it follows that $p$ occurs positively in both $\phi$ and $\psi$. The case where
$\neg p\in \tau_{\textsc{l}}$ and $p\in \tau_{\textsc{r}}$ is similar, except that we choose $\vartheta_\tau = \neg p$. By analogous reasoning, in this case, $p$ must occur negatively in both $\phi$ and $\psi$.

2. For some $D\in\{L,R\}$, $\tau_D$ contains a formula of the form $\Diamond\chi$, and  every type $\tau'=(\tau'_{\textsc{l}},\tau'_{\textsc{r}})$ satisfying $\tau\Rightarrow \tau'$ and $\chi\in\tau'_D$ has already been eliminated earlier on in the elimination sequence. In this case, we proceed as follows.
If $D=L$, we take

\[ \vartheta_\tau ~~~=~~~ \Diamond
   \bigvee_{\text{\begin{tabular}{c}$X$ an \textsc{l}-type with \\ $\chi\in X$ and $\tau_{\textsc{l}}\Rightarrow X$\end{tabular}}}
   \bigwedge_{\text{\begin{tabular}{c}$\tau'_{\textsc{r}}$ an \textsc{r}-type with \\ $\tau_{\textsc{r}}\Rightarrow Y$\end{tabular}}}
   \vartheta_{(X,Y)}
\]
We need to show that (i) $\models\bigwedge \tau_{\textsc{l}}\to \vartheta_\tau$ and (ii) $\models\vartheta_\tau\to\neg(\bigwedge\tau_{\textsc{r}})$.

For (i), let $M,w\models\bigwedge\tau_{\textsc{l}}$,
where $M=(W^M,R^M,V^M)$.
Then $M,w\models\Diamond\chi$.
Let $v$ be a witnessing successor, i.e., $(w,v)\in R^M$ and $M,v\models\chi$.
Let $X=\tau_{\textsc{l}}(v)$ be the \textsc{l}-type of $v$. Note
that, by construction, $\tau_{\textsc{l}}\Rightarrow X$.
By induction hypothesis, 
we have that $\models(\bigwedge X)\to\vartheta_{(X,Y)}$ for all
\textsc{r}-types $Y$ such that $\tau'=(X,Y)$ got eliminated previously.
It follows that \[M,v\models\bigwedge_{\text{$Y$ an \textsc{r}-type with $\tau_{\textsc{r}}\Rightarrow Y$}}\vartheta_{(X,Y)}\]
Hence, $M,w\models\vartheta_\tau$.

For (ii), by contraposition, let $M,w\models\bigwedge\tau_{\textsc{r}}$. We need to show
that $M,w\not\models\vartheta_\tau$. Let $v$ be any successor of $w$, and let 
$Y=\tau_{\textsc{r}}(v)$ be the \textsc{r}-type of $v$. Observe that $\tau_{\textsc{r}}\Rightarrow Y$.
It follows by the induction hypothesis that $M,v\not\models\vartheta_{(X,Y)}$, for all $X$ such that 
$(X,Y)$ got eliminated.
Hence, $M,w\not\models\vartheta_\tau$.

This concludes the argument for $D=L$.
If $D=R$, the argument is analogous, except that we reason dually, and we now take
\[ \vartheta_\tau ~~~=~~~ \Box
   \bigvee_{\text{\begin{tabular}{c}$X$ an \textsc{l}-type with \\ $\tau_{\textsc{l}}\Rightarrow X$\end{tabular}}}
   \bigwedge_{\text{\begin{tabular}{c}$Y$ an \textsc{r}-type with \\ $\chi\in Y$ and $\tau_{\textsc{r}}\Rightarrow Y$\end{tabular}}}
   \vartheta_{(X,Y)}
\]
We omit the details, which are straightforward. 
\end{proof}

From this theorem we get, as a corollary, the Lyndon interpolation theorem.

\begin{corollary} [Lyndon interpolation] 
If $\models\phi\to\psi$, then

\[ \vartheta = 
   \bigvee_{\text{\begin{tabular}{c}$X$ an \textsc{l}-type \\ with $\nnf(\phi)\in X$\end{tabular}}}
   \bigwedge_{\text{\begin{tabular}{c}$Y$ an \textsc{r}-type \\ with $\nnf(\neg\psi)\in Y$\end{tabular}}}
   \vartheta_{(X,Y)}
   \]

is a Lyndon interpolant for $\phi\to\psi$. 
\end{corollary}

\begin{proof}
If $\phi\land\neg\psi$ is unsatisfiable, then the type elimination sequence  ends in a quasi-model
that does not contain a type $\tau=(\tau_{\textsc{l}},\tau_{\textsc{r}})$ with $\phi\in \tau_{\textsc{l}}$ and
$\nnf(\neg\psi)\in\tau_{\textsc{r}}$. Therefore, every such type gets eliminated.
This shows that the above formula $\vartheta$ is indeed well defined.

To see that $\models\phi\to\vartheta$, let $M,w\models\phi$.
Let $X=\tau_{\textsc{l}}(w)$ be the \textsc{l}-type of $w$. By construction, $M,w\models\bigwedge X$. It follows from Theorem~\ref{thm:main} that,
for all \textsc{r}-types $Y$ containing $\nnf(\neg\psi)$, because $(X, Y)$ got eliminated, 
$M,w\models \vartheta_{(X,Y)}$. It follows that $M,w\models\vartheta$.

To see that $\models\vartheta\to\psi$, by contraposition, let $M,w\not\models\psi$.
Let $Y=\tau_{\textsc{r}}(w)$ be the \textsc{r}-type of $w$. By construction, $M,w\models\bigwedge Y$. It follows by Theorem~\ref{thm:main} that $M,w\not\models\vartheta_{(X,Y)}$ for all \textsc{l}-types $X$ containing $\nnf(\phi)$.
Therefore, $M,w\not\models\vartheta$. 
\end{proof}

By a careful inspection of the above procedure it can be shown that this 
yields a method for constructing interpolants in exponential time, provided
that the interpolants are represented succinctly using a DAG-style representation
of formulas.

The method described here generalizes to the guarded-negation fragment~\cite{Benedikt2015:effective} and yields an effective procedure for interpolant construction that is optimal in terms of formula length and running time, in that context.
It seems likely that the above method can be adapted to other modal logics, as well as extensions of the basic modal language, that admit filtration.
%
%
Also, unlike the model-theoretic approach (cf.~Section~\ref{sec:model}), the technique presented here is not inherently restricted to fragments of \emph{first-order} logic. We leave it as an open question
whether the present technique be used to prove a general interpolation result for a larger class of modal logics that admit filtration.

In summary:

\pro{The present technique is likely to generalize to other modal logics that admit filtration. It also generalizes to richer languages such as the guarded-negation fragment of first-order logic.}

\pro{The present proof technique yields an algorithm to compute Lyndon interpolants for \logic{K} in singly exponential time and of singly-exponential size, if a DAG-representation is allowed.}

\pro{The present proof technique can also be adapted for deciding the existence of Craig interpolants for certain extensions of modal logic that lack Craig interpolation~\cite{Jung2022:more}.}

\con{This proof technique does not extend in any obvious way to prove uniform interpolation.}




\section{Sixth proof (via modal algebras and duality theory)}
\label{sec:algebra}


Many notions of modal logic can be translated into algebraic concepts in terms of modal algebras. In particular, as we will see below, a modal logic $L$ has the Craig interpolation property iff the class of $L$-algebras---modal algebras corresponding to this logic---have the superamalgamation property. Thus, for showing that the basic modal logic $\logic{K}$ has the Craig interpolation, it is sufficient to show that its algebraic models, modal algebras, have the superamalgamation property. However, in order to show that the class of modal algebras has the superamalgamation property, the easiest is to work with duals of modal algebras---modal spaces. Using the spatial  intuition of modal spaces we can construct algebras that ensure that modal algebras have the superamalgamation property. As a corollary, we obtain that the basic modal logic $\logic{K}$ has Craig interpolation. In this section, 
we spell out this approach. 

\subsection{Algebraic semantics of modal logic}

Recall that a \emph{Boolean algebra} is a set $A$ equipped
with binary operations $\vee$ and $\wedge$, a unary operation $-$, and designated elements $0$ and $1$, satisfying the commutative and associative laws for $\vee$ and $\wedge$, distribution of $\vee$ over $\wedge$ and vice versa, and a few other equations.
A \emph{modal algebra} is a Boolean algebra with an
additional unary operation $\Diamond$, such that $\Diamond$ commutes with finite joins (i.e., $\Diamond (a\vee b)=\Diamond a\vee\Diamond b$) and $\Diamond 0 = 0$. We will denote such 
modal algebras by $\mathfrak{A}=(A,\vee,\wedge,-,0,1, \Diamond,)$ or
just $\mathfrak{A}=(A,\Diamond)$.
Note also that on each Boolean algebra we can also define  the Boolean implication $a\to b = - a \vee b$ and on a modal algebra we can define the modal Box: $\Box a = - \Diamond - a$.
A \emph{homomorphism} $h:\mathfrak{A}_1\to\mathfrak{A}_2$, as usual in universal algebra, is a map from the underlying set of $\mathfrak{A}_1$ to the underlying set of $\mathfrak{A}_2$ that commutes with 
the operations (i.e., $\vee$, $\wedge$, $-$, and $\Diamond$) and that maps the 
designated elements of $\mathfrak{A}_1$ (i.e., the 0 and 1) to the corresponding designated elements of $\mathfrak{A}_2$.

Recall from Section~\ref{sec:basic} that a \emph{normal modal logic} is a set $L$ of modal formulas
that includes all substitution instances of propositional tautologies as well as all formulas of the form $\Box(\chi_1\to\chi_2)\to(\Box\chi_1\to\Box\chi_2)$ 
and is closed under modus ponens, necessitation, and substitution; and that $\logic{K}$ is the least normal modal logic.

Let $(A, \Diamond)$ be a modal algebra. A \emph{valuation} on $A$ is a map $v: \mathsf{Prop}\to A$, where $\mathsf{Prop}$ is the set of propositional variables. 
This map can be extended to the set of all modal formulas in an obvious way. We say that a modal formula $\varphi$ is \emph{valid} on a modal algebra $(A, \Diamond)$ (written $(A, \Diamond)\models \varphi)$ if $v(\varphi) = 1$ for any valuation $v$ on $A$. A formula $\varphi$ is \emph{valid} on a class of algebras $\mathcal{C}$ (written $\mathcal{C}\models \varphi$) if 
$\varphi$ is valid on every algebra in $\mathcal{C}$.

A normal modal logic $L$ is sound and complete with respect to a class $\mathcal{C}$ of modal algebras if 
for each modal formula $\varphi$ we have  
$$
\vdash_L \varphi\ \mbox{iff}\ \mathcal{C}\models \varphi
$$
For any normal modal logic $L$, an \emph{$L$-algebra} is a modal algebra that validates all the formulas in $L$. In particular, for the basic modal logic $\mathsf{K}$ the class of $\mathsf{K}$-algebras coincides with the class of all modal algebras. For the next result we refer to e.g., \cite[Chapter 5]{BlackburnDeRijkeVenema} or \cite[Chapter 7]{CZ97}. 

\begin{theorem}[Algebraic completeness]\label{A-comp}
Every normal modal logic $L$ is sound and complete with respect to the class of $L$-algebras. 
\end{theorem}
 
We note, on the other hand, that there exist modal logics that are Kripke incomplete \cite[Chapter 6]{CZ97}, \cite[Chapter 4]{BlackburnDeRijkeVenema}. In fact, there are continuum many such logics \cite{Lit02}.

In order to state the next result we need to introduce 
the notion of \emph{amalgamable} and \emph{superamalgamable} classes of algebras.
Recall that an \emph{embedding} is an injective homomorphism. 

\begin{definition}\label{def: amalg} Let $\mathcal{C}$ be a class of algebras. We say that $\mathcal{C}$ is \emph{amalgamable} if for each $\mathfrak{A}_0$, 
$\mathfrak{A}_1, \mathfrak{A}_2\in \mathcal{C}$  such that  $\mathfrak{A}_0$ is embedded into  $\mathfrak{A}_1$ via a homomorphisms $h_1$ and 
$\mathfrak{A}_0$ is embedded into $\mathfrak{A}_2$ via a homomorphism $h_2$, there is an algebra $\mathfrak{A}_3\in \mathcal{C}$ such that 
$\mathfrak{A}_1$ is embedded into $\mathfrak{A}_3$ via a homomorphism $g_1$ and $\mathfrak{A}_2$ is embedded into $\mathfrak{A}_3$ via a homomorphism $g_2$ 
such that for each $a\in \mathfrak{A}_0$ we have $g_1(h_1(a)) = g_2(h_2(a))$. 
\begin{center}
\begin{tikzcd}
\mathfrak{A}_0 \arrow[d, "h_2"] \arrow[r, "h_1"] & \mathfrak{A}_1   \arrow[d, dashed, "g_1"]\\
\mathfrak{A}_2      \arrow[r, dashed, "g_2"]       &    \mathfrak{A_3}      
\end{tikzcd} 
\end{center}
\end{definition}

Note that the above definition of amalgamation applies to any class of universal algebras. 
If the conditions of Definition~\ref{def: amalg} are satisfied, the triple $\mathfrak{A}_0, \mathfrak{A}_1, \mathfrak{A}_2$ is called an \emph{amalgamable triple}. The algebra $\mathfrak{A}_3$ is called the \emph{amalgam} of the diagram or the \emph{algebra amalgamating the diagram}.

The next  definition of  superamalgamation presumes that the algebras in the class are lattice-based with  $\leq$ being  a lattice induced partial order. 
We note that for modal algebras $a\leq b$ iff $a\wedge b = a$ iff $a\vee b = b$ iff $a\to b = 1$. We will denote by $\leq_i$ the lattice order of an algebra $\mathfrak{A}_i$. 


\begin{definition} Let $\mathcal{C}$ be a class of modal algebras. 
We say that $\mathcal{C}$ is \emph{superamalgamable}  if it is amalgamable and, in addition, if for each amalgamable triple    
$\mathfrak{A}_0$, 
$\mathfrak{A}_1$, $\mathfrak{A}_2\in \mathcal{C}$ such that   $\mathfrak{A}_0$ is embedded into  $\mathfrak{A}_1$ via $h_1$ and 
$\mathfrak{A}_0$ is embedded into $\mathfrak{A}_2$ via $h_2$, and for each $a\in \mathfrak{A}_i$, $b\in \mathfrak{A}_j$, 
$\{i, j\} = \{1, 2 \}$ we have 
$$
g_i(a) \leq_3 g_j(b)\ \mbox{implies} \ \exists c\in \mathfrak{A}_0\  \mbox{such that } a\leq_i h_i(c)\ \mbox{and}\ h_j(c)\leq _j b.
$$
\end{definition}

Clearly, every class $\mathcal{C}$ which is supermalgamable, is also amalgamable. The next theorem connects superamalgamability  with the Craig interpolation property.


\begin{theorem}\label{CI = SA}
A normal modal logic $L$ has the Craig interpolation property iff the class of all $L$-algebras is superamalgamable. 
\end{theorem}

\begin{proof}
For the proof we refer to 
\cite[Chapter 14]{CZ97}, \cite[Chapter 7]{MAksGab05}, and 
\cite[Section 6.3]{Ven07}. We only sketch here why  superamalgamation of the class  of all $L$-algebras  entails that a modal logic $L$ has the interpolation property. We rely on the construction of the Lindenbaum--Tarski algebra for a modal logic $L$. The Lindenbaum--Tarski algebra $F_L(\vec{p})$ is the algebra of formulas in variables $\vec{p}$ moded out by $L$-equivalence. Moreover, it is well known that $\vdash_L \varphi(\vec{p})\to \psi(\vec{p})$ iff $F_L(\vec{p}) \models \varphi(\vec{p})\to \psi(\vec{p})$ iff $[\varphi(\vec{p})]\leq [\psi(\vec{p})]$ in $F_L(\vec{p})$, where $[\varphi]$ and $[\psi]$ are the elements of $F_L(\vec{p})$ corresponding to formulas $\varphi$ and $\psi$, respectively.

Now suppose $\vdash_L \varphi(\vec{p}, \vec{q})\to \psi(\vec{q}, \vec{r})$. 
Let $F_L(\vec{q})$, $F_L(\vec{p}, \vec{q})$ $F_L(\vec{q}, \vec{r})$ be the Lindenbaum--Tarski algebras of $L$ (see e.g., \cite{BlackburnDeRijkeVenema, Kra99, CZ97}) in the variables $\vec{q}$, and  $\vec{p}, \vec{q}$ and $\vec{q}, \vec{r}$, respectively. There is a natural embedding of $F_L(\vec{q})$ into $F_L(\vec{p}, \vec{q})$ and of $F_L(\vec{q})$ into $F_L(\vec{q}, \vec{r})$. Then it can be shown  that  the algebra $\mathfrak{A}_3$ amalgamating this diagram, is a subalgebra  of $F_L(\vec{p}, \vec{q}, \vec{r})$---the Lindenbaum--Tarski algebra of $L$ in the variables $\vec{p}$,   $\vec{q}$ and $\vec{r}$. We omit the details.


Finally, by just unfolding the superamalgamation property we obtain that there is a formula $\chi(\vec{q})$ in the Lindenbaum-Tarski algebra $F_L(\vec{q})$ (formally speaking an equivalence class of formulas in this algebra), such that $\vdash_L \varphi(\vec{p}, \vec{q})\to \chi(\vec{q})$
and $\vdash_L \chi(\vec{q})\to \psi(\vec{q}, \vec{r})$, entailing that $L$ has the interpolation property. 
\end{proof}


Thus, in order to show that $\logic{K}$ has the interpolation property, it is sufficient to show that the class of all modal algebras has the superamalgamation property. 
For this, we will make
use of duality theory,  a powerful technique that has many applications all over algebraic logic.
%

\subsection{Topological Duality}

In this section we briefly recall the basic definitions and facts underlying the topological duality for modal algebras. We start by reviewing Stone duality. For more information we refer to \cite{DP02, PJ82}.

A \emph{topological space} is a pair $(X, \tau)$ where $X$ is a set and $\tau$ is a collection of subsets of $X$ (called \emph{open sets}) satisfying the usual axioms: the empty set and $X$ itself are open; arbitrary unions of open sets are open; and finite intersections of open sets are open. A subset $C \subseteq X$ is \emph{closed} if its complement $X \setminus C$ is open. A space is \emph{compact} if every open cover has a finite subcover, and \emph{Hausdorff} if any two distinct points have disjoint open neighborhoods.
A \emph{Stone space} is a compact Hausdorff space that is \emph{totally separated}, meaning that distinct points are separated by a clopen (simultaneously closed and open) set. 

Stone spaces provide the topological side of \emph{Stone duality}, which establishes a correspondence between Boolean algebras and such spaces. Given a Boolean algebra $A$, its \emph{Stone space} (or \emph{Stone dual}) $A_\ast$ is the set of all \emph{ultrafilters} on $A$, that is, maximal proper filters $x \subseteq A$ satisfying $a \in x$ iff $\neg a \notin x$. The topology on $A_\ast$ is generated by basic open sets of the form
$$
  \widehat{a} = \{\, x \in A_\ast : a \in X \,\},
$$
for $a \in A$. Each $\widehat{a}$ is clopen, which ensures that $B_\ast$ is totally separated. Conversely, every Stone space $X$ gives rise to a Boolean algebra $\mathsf{Clop}(X)$ of clopen subsets, and these two constructions establish a dual equivalence between the categories of Boolean algebras and Stone spaces.

Here we will also like to highlight a useful separation property of Stone spaces, which we will use below in an essential way (see e.g., \cite{Eng77}). 

\begin{proposition}\label{Stone=>Normal}
    Let $X$ be a Stone space. For each pair of disjoint closed sets $F, G\subseteq X$, there exists a clopen set $C\subseteq X$ such that $F\subseteq C$ and $G\cap C = \varnothing$.
\end{proposition}

The crucial idea of the J\'onsson--Tarski duality, is that modal algebras can be represented as algebras of clopen sets of modal spaces---Stone spaces equipped with a binary relation. 
We refer to  \cite{BlackburnDeRijkeVenema, CZ97,Kra99,Ven07} for the basic theory of 
modal algebras, modal spaces, 
and the dual equivalence between modal algebras and modal spaces.

A \emph{modal space} (or \emph{descriptive Kripke frame}) is a pair $\mathfrak X=(X,R)$, where $X$ is a Stone space and $R$ is a binary relation on $X$ satisfying
$$
R[x]:=\{y\in X:xRy\}
$$
is closed for each $x\in X$ and
$$
R^{-1}[U]:=\{x\in X:\exists y\in U \mbox{ with } xRy\}
$$
is clopen for each clopen $U$ of $X$. A \emph{bounded morphism} (or \emph{p-morphism}) between two modal spaces is a continuous map $f$ such that $f(R[x])=R[f(x)]$. Let {\sf MS} be the category of modal spaces and bounded morphisms.

It is a well-known theorem in modal logic that {\sf MA} is dually equivalent to {\sf MS}. The functors $(-)_*:{\sf MA}\to{\sf MS}$ and $(-)^*:{\sf MS}\to{\sf MA}$ that establish this dual equivalence are constructed as follows. For a modal algebra $\mathfrak A=(A,\Diamond)$, let $\mathfrak A_*=(A_*,R)$, where $A_*$ is the Stone space of $A$  and 
$$
  x  R y \quad \text{iff} \quad \forall a \in A, \; (a \in y \Rightarrow \Diamond a \in x).
  $$
\noindent
We call $R$ the \emph{dual} of $\Diamond$. For a modal homomorphism $h:A\to B$, its dual $h_*:B_*\to A_*$ is given by $h^{-1}$. For a modal space $\mathfrak X=(X,R)$, let $\mathfrak X^*=(X^*,\Diamond_R)$, where $X^*$ is the Boolean algebra of clopens of $X$ and 
$$\Diamond_R(U)=R^{-1}[U].$$ 

For a  bounded morphism $f:X\to Y$, its dual $f^*:Y^*\to X^*$ is given by $f^{-1}$, the map that maps each set $U$ to its pre-image $f^{-1}(U)$. Moreover, $f$ is surjective  iff $f^*$ is injective and conversely $f$ is injective iff $f^*$ is surjective. 

The \emph{canonical frame} of a modal logic $L$ (or of its Lindenbaum--Tarski algebra) is a special instance of this general dual construction. The underlying set of the canonical frame consists precisely of the ultrafilters of the Lindenbaum-Tarski algebra (or, equivalently, maximally $L$-consistent sets), and the accessibility relation is defined as above. Thus, the familiar canonical model of $L$ in modal logic is just the J\'onsson--Tarski dual of its Lindenbaum-Tarski algebra. 

 Note that one can also construct the powerset algebra from a modal space, or any Kripke frame $(X, R)$ in general. The powerset algebra of $(X, R)$ is a the pair $(\mathcal{P}(X), \Diamond)$, where $\mathcal{P}(X)$ is the powerset Boolean algebra and $\Diamond_R$ is defined as $\Diamond_R(U)=R^{-1}[U]$ for each $U\in \mathcal{P}(X)$. It is easy to see that the powerset algebra of every Kripke frame is a modal algebra. 

This topological perspective allows one to transfer algebraic properties such as the \emph{superamalgamation property} into the corresponding dual notions, as we will see below.
 
\subsection{Superamalgamation via duality}

 We will now show how the above duality techniques can be put to use for proving that the modal logic $\mathsf{K}$ has the Craig interpolation property. 

\begin{theorem}\label{thm: sa}
    The class of all modal algebras is superamalgamable.
\end{theorem}

\begin{proof}
Consider the diagram of modal algebras $\mathfrak{A}_1$, $\mathfrak{A}_2$ and $\mathfrak{A}_3$, where
$h_1$ and $h_2$ are embeddings

\begin{center}
\begin{tikzcd}
\mathfrak{A}_0 \arrow[d, swap, "h_2"] \arrow[r, "h_1"] & \mathfrak{A}_1  \\
\mathfrak{A}_2              &           
\end{tikzcd} 
\end{center}

Let $(X_i, R_i)$ be the modal space dual to $\mathfrak{A}_i$. Let $f_1: X_1 \to X_0$ be the surjective bounded morphism  dual to $h_1$ and $f_2: X_2 \to X_0$ be the surjective bounded morphism dual to $h_2$. Then we get the dual diagram 

\begin{center}
\begin{tikzcd}
X_0   & X_1  \arrow[l, swap, "f_1"] \\
X_2  \arrow[u, "f_2"]             &           
\end{tikzcd} 
\end{center}

Let $$Y = \{(x, y): x\in X_1, y\in X_2, f_1(x) = f_2(y)\}$$

we also let $(x, y) R (x', y')$ if $x R_1 x'$ and $y R_2 y'$.

\begin{center}
\begin{tikzcd}
X_0   & X_1  \arrow[l, swap, "f_1"] \\
X_2  \arrow[u, "f_2"]             &    Y     \arrow[l, "\pi_2"] \arrow[u, swap, "\pi_1"]    
\end{tikzcd} 
\end{center}

If $(X_0, R_0)$, $(X_1, R_1)$ and $(X_2, R_2)$ were just relational structures this would do the job---$(Y, R)$ would co-amalgamate the above diagram, where $\pi_1$ and $\pi_2$ are projections from $Y$ to $X_1$ and $X_2$ respectively.  However, $(Y, R)$ has to be a modal space and while $Y$ is a closed subset of $X_1\times X_2$, and hence a Stone space, the space $(Y, R)$ may not be a modal space---it would be if $Y$ were a generated subframe, i.e., also closed under the product relation $R_1\times R_2$. 

However, one can still use $Y$ to prove the superamalgamation property. We will take the powerset algebra $(\mathcal{P}(Y), \Diamond_R)$. We let $\mathfrak{A}_3 = (\mathcal{P}(Y), \Diamond_R)$ and we show that $\mathfrak{A}_3$ superamalgamates the diagram. Let $\pi:Y \to X_1$ and $\pi_2:Y \to X_2$ be the two projection maps defined by 
$\pi_1(x, y) = x$ and $\pi_2(x, y) = y$. 

\begin{center}
\begin{tikzcd}
\mathfrak{A}_0 \arrow[d, swap, "h_2"] 
\arrow[r, "h_1"] & \mathfrak{A}_1   \arrow[d, "g_1"]\\
\mathfrak{A}_2      \arrow[r, swap, "g_2"]       &      \mathfrak{A}_3
\end{tikzcd} 
\end{center}

Keeping in mind that  $\mathfrak{A}_i$ is isomorphic to $(\mathsf{Clop}(X_i), \Diamond_{R_i})$ for 
$i=1, 2, 3$,
we also have maps $g_1: \mathfrak{A}_1\to  \mathfrak{A}_3$ and $g_2: \mathfrak{A}_2\to  \mathfrak{A}_3$ defined by 
$g_1 (U) = \pi_1^\ast (U) = \pi_1^{-1} (U)$ and $g_2 (U) = \pi_2^\ast (U) = \pi_2^{-1} (U)$. 

 We thus have the diagram:

\begin{center}
\begin{tikzcd}
(\mathsf{Clop}(X_0), \Diamond_{R_0}) \arrow[d, swap, "f^\ast_2"] \arrow[r, "f^\ast_1"] & (\mathsf{Clop}(X_1), \Diamond_{R_1})   \arrow[d, "\pi^\ast_1"]\\
(\mathsf{Clop}(X_2), \Diamond_{R_2})     \arrow[r, swap, "\pi^\ast_2"]       &     (\mathcal{P}(Y), \Diamond_R)
\end{tikzcd} 
\end{center}

We will now verify the superamalgamation property. Our proof will be topological in nature and it connects the existence of a superamalgam and hence of the Craig interpolant with the topological separation property  of Stone spaces highlighted in Proposition~\ref{Stone=>Normal}---in 
Stone spaces disjoint closed sets can be separated by a clopen set. 
We will show that the existence of this clopen set dually corresponds to the existence of the Craig interpolant.  

We first note that, since  $f_1(\pi_1(x, y)) = f_1(x) = f_2(y) = f_2(\pi_2(x, y))$, it is easy to see that 
for each $U\in \mathsf{Clop}(X_0)$ we have $\pi_{1}^{-1}(f_1^{-1}(U))=\pi_{2}^{-1}(f_2^{-1}(U))$, that is, $\pi_{1}^\ast(f_1^\ast(U))=\pi_{2}^\ast(f_2^\ast(U))$. This shows that modal algebras enjoy the amalgamation property. 

Now assume that  
$\pi^\ast_i(a) \leq \pi^\ast_j(b)$. for some $i, j = 1,2$. Without the loss of generality we may  assume that $i=1$ and $j=2$. The other case is  analogous. So let $\pi^\ast_1(U) \leq \pi^\ast_2(V)$
for some $U\in \mathsf{Clop}(X_1)$ and $V\in \mathsf{Clop}(X_2)$. This means that $\pi^{-1}_1(U) \leq \pi^{-1}_2(V)$. Note that $U$ and $X_3\setminus V$ are clopen and hence closed sets. As every continuous map between compact Hausdorff spaces is closed (i.e., the direct image of a closed set is closed) \cite{Eng77}, 
$f_1[U]$ and $f_2[X_2\setminus V]$ are closed subsets of $X_0$. We claim that they are also disjoint. Suppose not. Then there is $x\in 
f_1[U]\cap f_2[X_2\setminus V]$. 
This yields that there are $y\in U$ and $z\in X_3\setminus V$ such that $f_1(y) = f_2(z)$. But then $(y, z)\in Y$ and $(y, z) \in \pi_1^{-1}(U)$. Hence, by assumption, $(y,z)\in \pi_{2}^{-1}(V)$; but then $z\in V$, which contradicts the fact that $z\notin V$.


By Proposition~\ref{Stone=>Normal}, there exists a clopen set $C$ such that $f_1[U]\subseteq C$ and  $C\cap f_2[X_2\setminus V] = \varnothing$. We will now show that  $U\subseteq_1 f^\ast_1(C)$ and $f^\ast_2(C) \subseteq_2 V$. 
If not, then there is $x\in U$ such that $x\notin f_2^{-1}(C)$. Since $f_1[U]\subseteq C$, we have that $U\subseteq f^{-1}_1f_1[U]\subseteq f_1^{-1}(C)$. On the other hand, $C\cap f_2[X_2\setminus V] = \varnothing$ implies that $f_1^{-1}(C)\cap f_2^{-1}(f_2[X_2\setminus V]) = \varnothing$. This yields that  
$f_2^{-1}(C)\cap (X_2\setminus V)) = \varnothing$ and so $f_2^{-1}(C)\subseteq V$, which finishes the proof (via duality) that the class of modal algebras are superamalgamble. 
\end{proof}

Thus we obtain the following direct corollary to Theorems~\ref{A-comp}, \ref{CI = SA} and \ref{thm: sa}.

\begin{corollary}
The modal logic $\logic{K}$ has the Craig interpolation property. 
\end{corollary}




The method described here 
is algebraic and uses duality between modal algebras and modal spaces, it also uses essentially topological intuitions and topological properties of Stone spaces. In fact, it shows that the interpolant corresponds  topologically to the clopen set obtained via a separation property, as a ``separator'' of two disjoint closed sets. 

The proof of Theorem~\ref{thm: sa} shows that for a modal logic $L$ and the amalgamable triple of $L$-algebras $\mathfrak{A}_0$ $\mathfrak{A}_1$ and $\mathfrak{A}_2$, if the algebra $\mathfrak{A}_3 = (\mathcal{P}(Y), \Diamond_R)$, constructed in the proof of the theorem, is an $L$-algebra, then the class of $L$-algebras is superamalgamable and hence the logic $L$ has the Craig interpolation property. Note that the frame $(Y, R)$ is very similar to the bisimulation product constructed in the model-theoretic proof of Section~\ref{sec:model}. Also,  similarly to the model theoretic case, if a logic $L$ is axiomatized by modal formulas whose first-order correspondents are universal Horn sentences (see Section~\ref{sec:model}), then it is easy to see that the algebra $\mathfrak{A}_3 =(\mathcal{P}(Y), \Diamond_R)$ is an $L$-algebra, and hence $L$ has the Craig interpolation property (see also \cite[Section 6.3]{Ven07}). Thus, the two proofs---algebraic via duality and the model theoretic one---can be seen as two sides of the same coin.

\pro{The present proof technique is semantic in nature. The algebraic approach towards interpolation via  amalgamation and superamalgamation apply to many systems of non-classical logic admitting algebraic semantics, even if they do not possess Kripke semantics (see \refchapter{chapter:nonclassical} and \refchapter{chapter:algebra} in this volume).}

\pro{As the semantic condition of superamalgamation is necessary and sufficient for the logic to have the Craig interpolation, this method could also be used to show negative results, that is, that a given  logic (or a class of logics) does not have the interpiolation property. For example, using the algebraic and duality techniques Maksimova \cite{Maks1977, Maks1979} showed that there are exactly $7$ consistent superintuitionistic logics and that among the extensions of the modal logic $\mathsf{S4}$ there are at most $37$  logics with the interpolation property (see also \refchapter{chapter:nonclassical} in this volume).}


\con{The interpolants constructed in this section are highly non-constructive. Thus, the method does not provide an algorithm for finding the interpolants. }

\section{Lower bounds}
\label{sec:lowerbounds}

As we have seen, some Craig interpolation techniques (such as the syntactic approach using $\nabla$-normal form) yield interpolants of  exponential formula length, and some
others (such as the approach using quasi-models) yield interpolants that, although 
not necessarily of exponential length,  admit a DAG-representation
of singly exponential size.%
\footnote{It is not inconceivable that a more fine-grained  analysis would show that the approaches that yield an exponential size DAG representation in fact yield an exponential length formula.}
In this section, we present a matching size lower bound. Specifically, 
the following result, due to Frank Wolter \emph{(personal communication)}, shows that Craig interpolants for modal implications are in the worst case necessarily exponentially large, even when the size of a formula is measured by its DAG-representation. Note that 
no such unconditional lower bound is known for propositional logic (cf.~\refchapter{chapter:predicate}).

\begin{theorem} 
There is a sequence of valid modal implications $(\phi_n\to\psi_n)_{n=1,2,\ldots}$, such that 
\begin{enumerate}
    \item $|\phi_n|$ and $|\psi_n|$ are bounded by a polynomial function in $n$,
    \item Every Craig interpolant $\vartheta$ for $\phi_n\to\psi_n$ has DAG-size $|\vartheta|_{\text{DAG}}\geq 2^n$.
\end{enumerate}
\end{theorem}

\begin{proof}
   Take the following modal formulas with proposition letters $\{p_1, \ldots, p_n,s,q_1,\ldots, q_n\}$:
   \[
   \phi_n = \Diamond s\land \bigwedge_{i=1\ldots n}((p_i\to \Box(s\to p_i))\land (\neg p_i\to \Box(s\to \neg p_i)))
   \]
   \[
   \psi_n = \Big(\bigwedge_{i=1\ldots n}
((p_i \to \Box q_i) \land (\neg p_i \to \Box\neg q_i))\Big) \to \Diamond \bigwedge_{i=1\ldots n} (p_i\leftrightarrow q_i)
   \]
    \[
    \chi_n = \bigvee_{X\subseteq\{1, \ldots, n\}}(\tau_X\land\Diamond\tau_X)
    \text{ ~~~~ where $\tau_X=\bigwedge_{i\in X}p_i\land\bigwedge_{i\not\in X}\neg p_i$.}
    \]
Note that the formulas $\phi_n, \psi_n$ are indeed of size polynomial in $n$. It is also easy to see that $\phi_n\to\psi_n$ is a valid modal implication and that $\chi_n$ is a Craig interpolant for it. Therefore, it remains only to establish the following two claims.

\medskip\par\noindent\emph{Claim 1: }
Every Craig interpolant for $\models\phi_n\to\psi_n$ is equivalent to $\chi_n$

\medskip\par\noindent\emph{Proof of claim:} Let $\theta_n$ be a Craig interpolant for $\phi_n\to\psi_n$.
We must show that $\theta_n$ is equivalent to $\chi_n$. We do this by showing that
every pointed Kripke model satisfying $\chi_n$ can be expanded so that it satisfies $\phi_n$, while each pointed Kripke model satisfying $\neg\chi_n$ can be expanded so that it satisfies $\neg\psi_n$. Since $\chi_n$ lies sandwiched between $\phi_n$ and $\psi_n$,
it then follows that $\chi_n$ and $\theta_n$ are equivalent.

More precisely, let $M,w$ be any pointed Kripke model
in the propositional signature $\{p_1, \ldots, p_n\}$. If $M,w\models\chi_n$, then
we can expand $M,w$ into a model $M',w$ over
the extended signature $\{p_1, \ldots, p_n,s\}$ by making $s$ true (only) at some $R$-successor of $w$ that satisfies the same proposition letters as $w$.
We then have that $M',w\models\phi_n$
and hence $M',w\models\theta_n$. Since
$s$ does not occur in $\theta_n$, it 
follows that $M,w\models\theta_n$.
If, on the other hand, $M,w\not\models\chi_n$, we can expand
$M,w$ into a model $M'',w$ over the
extended signature $\{p_1, \ldots, p_n,q_1, \ldots, q_n\}$ by making $q_i$
true everywhere if $p_i$ is true at $w$
and making $q_i$ true nowhere otherwise. In this way, we have that
$M'',w\not\models\psi_n$ and hence
$M'',w\not\models\theta_n$. Since $q_1, \ldots, q_n$ do not occur in $\theta_n$,
it follows that $M,w\not\models\theta_n$.

\medskip\par\noindent\emph{Claim 2: }
Every modal formula equivalent to $\chi_n$ has DAG-size at least $2^n$.

\medskip\par\noindent\emph{Proof of claim:}
Assume for the sake of a contradiction
that $\chi_n$ is equivalent to $\vartheta_n$
for some modal formula $\vartheta_n$ of DAG-size strictly less than $2^n$.
Consider the Kripke model $M_n = (W_n, R_n, V_n)$ with 
\begin{itemize}
    \item $W_n = \{w\}\cup\{v_X\mid X\subseteq\{1, \ldots, n\}\}\cup\{u_X\mid X\subseteq\{1, \ldots, n\}\}$
    \item 
$R_n = \{(w,v_X)\mid X\subseteq\{1, \ldots, n\}\}\cup\{(v_X,u_Y)\mid X,Y\subseteq\{1, \ldots, n\}\}$; and
\item 
$V_n(p_i)=\{v_X\mid i\in X\}\cup \{u_X\mid i\in X\}$.
\end{itemize}
By construction, we have that $M_n,w\models\Box\chi_n$, and for every  submodel $M'_n$ of $M_n$ obtained by deleting one or more
worlds of the form $u_X$, we have that 
$M'_n,w\not\models\Box\chi_n$.
At the same time, since $\vartheta_n$ has
DAG-size strictly smaller than $2^n$, it
has strictly less than $2^n$ subformulas, and therefore, 
there exists a world $u_X$ 
such that every subformula of $\vartheta_n$ true in $u_X$ is also true in $u_{X'}$ for some $X'\neq X$. It follows that deleting the world $u_{X}$ 
from the model $M_n$ does not affect
the truth of $\Box\vartheta_n$, and hence of $\Box\chi_n$, at $w$, 
a contradiction.
\end{proof}

The same argument is used in~\cite[Lemma 22]{Kuijper25:separation} to establish an exponential succinctness gap between modal logic and the two-variable fragment of first-order logic.

\section*{Acknowledgments}
\addcontentsline{toc}{section}{Acknowledgments}
The authors are very grateful to Roman Kuznets and Patrick Koopman for the careful reading of the drafts of this paper and for their many helpful comments and corrections. They also wish to thank Rodrigo Almeida, Konstantinos Papafilippou, and Frank Wolter for valuable discussions and suggestions.


\bibliography{admin/bib.bib}

\end{document}